\documentclass[a4paper, 11pt]{amsart}
\usepackage{color,array,tikz, multirow,cancel, enumerate,cite}
\usepackage{amssymb,amsmath,amsfonts,amsthm,euscript,mathrsfs}
\usepackage[english]{babel}
  \usepackage[pdftex]{hyperref}
\usepackage{epsf}
\csname @addtoreset\endcsname{equation}{section}
\newcommand{\SU}{\mathop{\rm SU}}
\newcommand{\SO}{\mathop{\rm SO}}
\newcommand{\Sp}{\mathop{\rm {}Sp}}
\newcommand{\Sl}{\mathop{\rm {}SL} }

\newtheorem{theorem}{Theorem}[section]
\newtheorem{lemma}[theorem]{Lemma}

\newtheorem{prop}[theorem]{Proposition}

\theoremstyle{definition}
\newtheorem{definition}[theorem]{Definition}
\theoremstyle{remark}
\newtheorem{remark}[theorem]{Remark}
\newtheorem{example}[theorem]{Example}

\newenvironment{notation}{\textbf{Notation.}}{ }

\begin{document}
\begin{titlepage}
\begin{flushright}
MPP-2012-128
\end{flushright}
 
\begin{center}
\vspace{3cm}
\baselineskip=18pt{\huge
Tate form and weak coupling limits in F-theory   \\
}
\vspace{2 cm}
{\Large  Mboyo Esole$^{\spadesuit,\heartsuit, 
 \clubsuit}$  and $\;\,$\Large Raffaele Savelli$^\diamondsuit$ }\\
\vspace{1 cm}
\center{
${}^\spadesuit$Department of Mathematics, \ Harvard University, Cambridge, MA 02138, U.S.A.\\
${}^\heartsuit$Jefferson Physical Laboratory, Harvard University, Cambridge, MA 02138, U.S.A.\\
${}^\clubsuit$Taida Institute for Mathematical Science, National Taiwan University, 
 Taipei, 
Taiwan.\\
$^\diamondsuit$ Max-Planck-Institut f\"ur Physik, F\"ohringer Ring 6, 80805 Munich, Germany.}
\end{center}

\vspace{1cm}
\begin{center}

{\bf Absract}
\vspace{.3 cm}
\end{center}
{\small
We consider the weak coupling limit of F-theory in the presence of non-Abelian gauge groups implemented using the traditional  ansatz coming from 
 Tate's algorithm. We classify the types of singularities that could appear in the weak coupling limit and explain their resolution. In particular,  the weak coupling limit of $\SU(n)$ gauge groups leads to an orientifold theory which suffers from conifold singulaties that do not admit a crepant resolution 
compatible with the orientifold involution. We present a simple resolution to this problem by introducing a new weak coupling regime that admits singularities compatible with both a crepant resolution and an  orientifold symmetry. We also comment on possible applications of the new limit to model building.
We finally discuss other  unexpected phenomena as for example the existence of several non-equivalent directions to flow from strong to weak coupling leading to different gauge groups. 
\vfill

\flushleft{$^\spadesuit$Email:\quad {\tt    esole at  math.harvard.edu}\\
$^\diamondsuit$Email:\quad{\tt    savelli at  mpp.mpg.de}}
}

\end{titlepage}
\addtocounter{page}{1}
 \tableofcontents{}
\newpage

\section{Introduction}

F-theory \cite{Vafa:1996xn} is a  powerful tool to engineer  gauge theories appearing in type IIB string theory and in the hetorotic string theory using the geometry of elliptic fibrations \cite{Morrison:1996na, Bershadsky:1996nh}. When the structure of an elliptic fibration is seen through the eyes of string theory, the constraints coming from physics lead to  surprising new mathematical results on the structure of elliptic fibrations \cite{FMW, KLRY, AE1,AE2,EY,  EFY, GM1, GM2,Braun:2011ux,
Park:2011ji, Fullwood:2011bf, 
Morrison:2011mb,
Morrison:2012js,Taylor:2012dr,Morrison:2012ei}.  Although F-theory is defined for elliptically fibered Calabi-Yau varieties, many of the mathematical results obtained in F-theory apply in a larger setup  without restrictions on the dimension of the elliptic fibration and without assuming the Calabi-Yau condition \cite{AE1,AE2,GM2}. 

One fascinating aspect of the F-theory  approach  is that  it provides a window to non-perturbative aspects of   type  IIB string theory.  
This is because the  elliptic fibration implements geometrically several non-perturbative aspects of S-duality. 
In type IIB string theory S-duality changes the value of the string coupling constant and can relate weak and strong couplings. 
Understanding the connection between the strongly coupled regime of F-theory and weakly coupled type IIB string theory is a central theme in F-theory \cite{Denef.LH}.

Sen has provided a beautiful description of a limit of an elliptic fibration which results in a type IIB orientifold theory at weak coupling \cite{
Sen:1996vd,Sen.Orientifold}. 
The orientifold theory is defined as the double cover $X\rightarrow B$ of the base over which the elliptic curve is fibered:
$$
X:\xi^2 =h.
$$
Sen's limit was originally defined  essentially for elliptic fibrations in the Weierstrass form. 
 A systematic way to describe the links between the orientifold theory obtained in  Sen's weak coupling limit and the geometry of the elliptic fibration were described in \cite{CDE}. 
One would like to be able to uplift a compactification in type IIB to F-theory in order to understand its strong coupling behavior.  
Such effort relies on exploiting properties of a given weak coupling limit.  
Progress on that has been made in the past few years based essentially on Sen's limit 
 \cite{Collinucci:2008zs}, \cite{Collinucci:2009uh,Blumenhagen:2009up}. 
Generalization of the weak coupling limit to other models of elliptic fibrations  were later obtained in \cite{AE2, EFY}. These generalizations are based on a geometric reformulation of  Sen's limit in terms of a transition from semi-stable to unstable singular fibers \cite{AE2}. 
These new limits illustrate among other things the non-uniqueness of the weak coupling limit in F-theory \cite{AE2, EFY}.  Different limits of the same F-theory model can lead to very different configurations at weak coupling.  
Reciprocally,  a given type IIB model can admit several nonequivalent  uplifts to F-theory. See for example \cite{Braun:2009wh} for explicit examples.

 Recently, F-theory has been used to obtain  
local models of Grand Unified Theories (GUTs), starting from the papers \cite{Donagi:2008ca,Beasley:2008dc,Beasley:2008kw,Hayashi:2009ge}. Global completions of GUT models have also been intensively studied \cite{Andreas:2009uf,Blumenhagen:2009yv,Grimm:2009yu, Marsano:2009gv, Cvetic:2010rq,Marsano:2011hv,Collinucci:2009uh, Tatar:2012tm, Marsano:2012yc}, with special focus on $\SU(5)$ configurations. See \cite{Weigand:2010wm} for a review on the subject and a more complete list of references.

In F-theory, there are specific ans\"atze that provide a realization of a particular  gauge group using a 
 generalized  
Weierstrass model with coefficients vanishing along the divisor of interest up to  certain multiplicities directly inspired by Tate's algorithm \cite{Bershadsky:1996nh, Katz:2011qp}. For that reasons, such ans\" atze are usually called {\em Tate form} in the F-theory literature. The restrictions to put a given elliptic fibration with a given singularity type into  a Tate form have been  analyzed recently in \cite{Katz:2011qp} where some obstructions have been noticed for certain groups ($SU(m)$ (with $6\leq m\leq 9$), $\Sp(3)$, $\Sp(4)$, $\SO(13)$, $\SO(14)$) or in the presence of certain matter representations (such as the 2-symmetric representation of $\SU(m)$). Normal forms for local equations for classical groups were also given in \cite{Katz:2011qp}. 
Donagi and Wijnholt have proposed a realization of Sen's limit for elliptic fibrations in Tate forms \cite{Donagi:2009ra}. 
For  a generalized Weierstrass model:
$$
y^2 z= x^3 + a_1 xyz + a_3 y z^2 =x^3 + a_2 x^2 z + a_4 x z^2 + a_6 z^3,
$$
Donagi and Wijnholt use the following  ansatz:
$$
\text{Donagi-Wijnholt}
\begin{cases}
a_3\rightarrow \epsilon a_3\\
a_4\rightarrow \epsilon a_4\\
a_6\rightarrow \epsilon^2 a_6
\end{cases}
$$
We will refer to that limit as the {\em DW limit}. It provides a simple realization of Sen's limit for a generalized Weierstrass model.
As we will see in section \ref{DWAnsatz}, the DW limit reduces to Sen's limit when the generalized Weierstrass equation is reduced to a short Weierstrass form.  
 The DW limit  is consistent with the Tate form of a nodal curve as it has $a_3=a_4=a_6=0$ in the limit $\epsilon\rightarrow 0$. 
  It applies in particular to models in which a gauge group is implemented using a Tate form.
For every nonzero value of $\epsilon$, it preserves the Tate form and therefore the gauge group. However as we take the limit $\epsilon\rightarrow 0$ the gauge group can change as we will see. 
 In the DW-limit, the double cover of the base used to describe the orientifold theory in type IIB takes the form 
\begin{equation}
X:\quad \xi^2 =a_1^2+ 4a_2.
\end{equation}
$X$ is nonsingular when $a_2=0$ describes a nonsingular divisor. However, the Tate form for groups other than $\Sp(n)$ will require $a_2$ to factorize and therefore to admit singularities. Typically, if  $a_1, a_2$ factorize as $a_2=\sigma^{m_2} b_{2,m_2}$ and $a_1=\sigma^{m_1} a_{1,m_1}$  for a given Tate form, we get 
 \begin{equation}
X:\quad \xi^2 =\sigma^{2 m_1}a_{1,m_1}^2+4 \sigma^{m_2} a_{2,m_2},
\end{equation}
which is singular when $m_2>0$. 

The presence of singularities is acceptable if one can obtain a crepant resolution  
$$\mu:\tilde{X}\rightarrow X$$
 compatible with the involution $\xi\mapsto -\xi$, which represents the $\mathbb{Z}_2$ orientifold symmetry in the weakly coupled type IIB picture. 
The crepant condition means that the first Chern class of the resolved variety $\tilde{X}$ is the pullback of the first Chern class of $X$: 
$$c_1(\tilde{X})=\mu^\star c_1(X).$$
 This is important in order to make sure that the string target space is still Calabi-Yau.  One would also want the resolution to be a double cover in order to still have an orientifold theory in the weak coupling limit. The compatibility condition requires that the involution of  $\tilde{X}$ naturally reduces to the one of $X$. We will call a resolution of a double cover an {\em admissible crepant resolution} if it is crepant and compatible with the double cover. 
Double covers and their admissible crepant resolutions are well studied in mathematics. We review the main results in  appendix \ref{AdmissibleResolutions}.  We  point out an important  subtlety involving the preservation of the  D7 tadpole after a resolution of singularities in section \ref{PhysicalProperties} and discuss it  further  in the conclusion (section \ref{Outlook}).

 In the case of unitary gauge groups, the DW weak coupling limit gives a singular orientifold theory which does not admit an admissible crepant resolution. Indeed, for the Tate form of a unitary gauge group we have $a_2=\sigma a_{2,1}$. This implies that the double cover admits  conifold-like singularities in codimension-3:
\begin{equation}\label{ConifoldGeometry}
X:\quad \xi^2 =a_1^2+ 4\sigma a_{2,1},
\end{equation}
where $\sigma=0$ defines the divisor on which the fiber $I^{s}_{n}$ is located and gives the $\SU(n)$ gauge group. 
Such a singularity is a serious problem, especially for GUT model building, as it makes the theory at weak coupling ill-defined for any kind of computation. 
These conifold singularites can be  avoided at the cost of working with elliptic fibrations  for which the conifold points $a_1=\sigma=a_{2,1}=0$ do not actually occur due to the intersection theory of the base \cite{Krause:2012he}. Such geometries can be algorithmically constructed starting from type IIB, as explained in \cite{Collinucci:2009uh}. 

In this paper,  we will consider a systematic way of solving the conifold problem in the weak coupling limit of $\SU(n)$ theories. If one is willing to give up the choice of having a gauge theory in eight dimensions with unitary gauge group, an immediate solution of the conifold problem would be to deform the singularity by adding on the right hand side of \eqref{ConifoldGeometry} a generic polynomial of the appropriate degree.  This breaks the unitary gauge group we started with by giving an expectation value to fields in its antisymmetric representation \cite{Donagi:2009ra}. However, this would only make us move from a split to a non-split singularity of the Calabi-Yau fourfold, by rendering the latter slightly more generic. As a result, monodromies of the fiber are introduced and the unbroken gauge group is of the symplectic type. In contrast, in this paper, we construct new limits that lead to  a double cover allowing for admissible crepant resolution. The limits we consider are specialization of the DW-limit. The idea behind is very simple: knowing the conditions to have an admissible crepant resolution (see appendix \ref{AdmissibleResolutions}), we improve the DW-limit in order to replace the conifold points by a better-behaved singularity. In particular, there is one choice of singularity which even allows us to preserve the unitary gauge group of the starting F-theory configuration. We consider the limit
$$
\begin{cases}
a_{2,1}\rightarrow \epsilon a_{2,1}+\sigma a_{2,2}\\
a_3\rightarrow \epsilon a_3\\
a_4\rightarrow \epsilon a_4\\
a_6\rightarrow \epsilon^2 a_6
\end{cases}
$$
 In this new limit, the conifold singularities are replaced by the singularities of  suspended pinch points:
$$
X: \xi^2=a_1^2 + \sigma^2 a_{2,2}.
$$
Such a double cover admits an admissible crepant resolution obtained by blowing-up the  locus $\xi=a_1=\sigma=0$ in codimension-two. It is also evident the presence of brane and image brane stacks when we put $\sigma=0$. We discuss this model at length and propose some applications of it to GUT model building. In particular we analyze the structure of matter curves before and after resolution and comment on specific suspended pinch point geometries which contain all the expected matter spectrum. We also raise a few questions, especially regarding the realization of Yukawa couplings, which we hope to address in the near future.

The paper is organized as follows. We begin in section \ref{Notations} by briefly introducing the geometries of Weierstrass elliptic fibrations and by fixing our notations, whereas in section \ref{SenWeakCoupling} we review the Sen weak coupling limit and its orientifold interpretation in type IIB string theory. In section \ref{DWAnsatz}
we present a general overview of the Donagi-Wijnholt limit and in subsection \ref{SolvingConifold} propose two alternatives which solve the conifold problem and show different physical features among themselves. Section \ref{GeneralitiesWCL} gives a broader overview on weak coupling limits and \ref{BraneContentDW} provides a detailed analysis of the brane content after taking DW limit for each kind of Kodaira singularity. Section \ref{resCY3spp} is instead devoted to a more-in-depth analysis of one of our alternative proposals, i.e. the suspended pinch point geometry: We blow-up the singularity and discuss the features of the resolved geometry from both a mathematical and a physical perspective. We draw our conclusions in section \ref{Outlook}. Finally some technical details are provided in the appendices: Some mathematical theorems relevant for our analysis are presented in appendix \ref{AdmissibleResolutions}, where we also propose an equivalent small resolution of the suspended pinch point geometry; In appendix \ref{BlowUpQuadricCone} a blow-up is given for the Quadric Cone geometry.

\section{Elliptic curves and Weierstrass models: a quick review}\label{Notations}

\begin{definition}[Elliptic curve and Weierstrass normal equation]
An elliptic curve over a field $K$ is an irreducible nonsingular projective algebraic curve of genus $1$  with a choice of a $K$-rational point, the origin of the group law.
  It  follows from Riemann-Roch theorem that  an elliptic curve over a field $K$  is isomorphic to a plane cubic curve,  cut in $\mathbb{P}^2_K$ by the following  {\em generalized Weierstrass form}:
\begin{equation}\label{wt}
E:\quad zy^2 + a_1 x y z+ a_3  y z^2 = x^3 + a_2 x^2 z + a_4 x z^2 + a_6 z^3,\quad a_i \in K.
\end{equation}
Geometrically, the marked point of the  Weierstrass form   of an elliptic curve is its intersection point with the line at infinity $z=0$, namely the point $[x:y:z]=[0:1:0]$, which is a point of inflection and the only point of infinity of the curve. 
The curve \eqref{wt} is called a Weierstrass normal  form  since (in characteristic different from 2 and 3) after the change of variables: 
$\wp=x+\frac{1}{12}(a_1^2 + 4 a_2)$, $\wp'=2y + a_1 x + a_3$
 it reduces to the traditional cubic equation satisfied by the Weierstrass $\wp$-function and its derivative:
$E_\Lambda:(\wp')^2=4 \wp^3- g_2 \wp-g_3$.
\end{definition}

\subsection{Elliptic fibration}
An elliptic fibration over a base $B$ can be seen as an elliptic curve over the function field of $B$. We will define elliptic fibrations by a Weierstrass model. 
The Weierstrass model over a base $B$ is written in a projective bundle $\mathbb{P}(\mathscr{E})\rightarrow B$ where $\mathscr{E}=\mathscr{O}_B\oplus \mathscr{L}^2\oplus \mathscr{L}^3$.  

\begin{notation}
The sheaf of regular functions of a variety  $B$ is denoted as usual by $\mathscr{O}_B$. 
We use the classical convention for  projective bundles $\mathbb{P}(\mathscr{E})$: At each point they are defined by the set of lines  of $\mathscr{E}$. We denote the tautological line bundle of the projective bundle $\mathbb{P}(\mathscr{E})$ by $\mathscr{O}_B(1)$. When the context is clear, we just denote it $\mathscr{O}(1)$.  
We denote by $\mathscr{O}(n)$ for $n>0$ the $n$th tensor product of $\mathscr{O}(1)$. Its dual is $\mathscr{O}(-n)$.
\end{notation}

The coefficients $a_i$ of \eqref{wt}  are sections of $\mathscr{L}^i$. The projective coordinates $[x:y:z]$ of this projective bundle are such that 
$x$ is a section of $\mathscr{O}(1)\otimes \pi^\star\mathscr{L}^2$, $y$ is a section of $\mathscr{O}(1)\otimes \pi^\star\mathscr{L}^3$ and $z$ is a section of 
$\mathscr{O}(1)$.  The elliptic fibration is a section of $\mathscr{O}(3)\otimes \pi^\star\mathscr{L}^6$ in the  bundle $\mathbb{P}(\mathscr{E})$.

The elliptic fibration $\varphi:Y\rightarrow B$ has vanishing Chern class if $ c_1(B)=c_1(\mathscr{L})$.
For most of this paper, except for the physical applications in section \ref{resCY3spp}, we will not need to impose the Calabi-Yau condition and we will also not restrict the dimension of the base. 
\subsection{Formulaire}
An elliptic curve given by a Weierstrass equation is singular if and only if its discriminant  $\Delta$ is zero.  If we denote by $\bar{K}$ the algebraic closure of $K$, two smooth elliptic curves are isomorphic over $\bar{K}$ if and only if they have the same $j$-invariant.  We recall  the formulaire of  Deligne and Tate which is   useful to express the discriminant $\Delta$,  the $j$-invariant and to reduce the Weierstrass equation into simpler forms:
\begin{align}\label{eq.formulaire}
b_2=& a_1^2+ 4 a_2,\  \    b_4= a_1 a_3 + 2 a_4 ,\   \   b_6 = a_3^2 + 4 a_6 , \    \   b_8 =b_2 a_6 -a_1 a_3 a_4 + a_2 a_3^2-a_4^2,\\
 c_4=& b_2^2 -24 b_4  (=12 g_2), \quad c_6 = -b_2^3+ 36 b_2 b_4 -216 b_6 (=216 g_3).\\ 
\intertext{The coefficients $b_i$ and $c_i$ are sections of $\mathscr{L}^i$. The  discriminant and the $j$-invariant are given by: }
  \Delta=&-b_2^2 b_8 -8 b_4^3 -27 b_6^2 + 9 b_2 b_4 b_6 (=g_2^3 -27 g_3^2),\qquad 
  j=\frac{c_4^3}{\Delta}(=1728 J).
\end{align}
These quantities are related by the following relations:
\begin{equation}
4 b_8 =b_2 b_6 -b_4^2 \quad \text{and}\quad 1728 \Delta=c_4^3 -c_6^2. 
\end{equation}
The variables $b_2, b_4, b_6$ are used to express the Weierstrass equation after completing the square in $y$ by a redefinition $y\mapsto y-\frac{1}{2}(a_1 x +a_3 z)$:
$$
z y^2= x^3 + \frac{1}{4}b_2 x^2 z+  \frac{1}{2} b_4 x z^2+  \frac{1}{4} b_6 z^3.
$$
The variables $c_2, c_4$ and $c_6$ are then obtained after eliminating the term in $x^2$ by the redefinition $x\mapsto x-\frac{1}{12} b_2 z $, in order to have the short form of the Weierstrass equation:
$z y^2=x^3 - \tfrac{1}{48} c_4 x z^2 -\tfrac{1}{864} c_6 z^3.$
We will use the following normalization of the short Weierstrass equation (obtained by introducing  $f= - \frac{1}{48} c_4$ and $ g=-\frac{1}{864} c_6 $): 
\begin{equation}
E:z y^2=x^3 + f x z^2 +g z^3, \quad \Delta=-16(4f^3 + 27 g^2) , \quad j=1728\frac{4f^3 }{4f^3 + 27 g^2}.
\end{equation}

\section{Weak coupling limit of elliptic fibrations}\label{SenWeakCoupling}

In Type IIB string theory, when taking into account the back-reaction of space-time-filling 7-brane sources, the string coupling $g_s$ is not a constant, but is varying along the internal space. It is defined by (the expectation value of) the exponential of the dilaton $\phi$. Together with the axion $C_0$, the dilaton  forms a complex scalar field called the {\em axio-dilaton}:
\begin{equation}
\tau=C_0+\mathrm{i}\, e^{-\phi},
\end{equation}
which transforms under the  S-duality group $\Sl(2,\mathbb{Z})$ as the complex modulus of  a torus under the action of the modular group:
\begin{equation}
\tau \mapsto \frac{a \tau + b}{c\tau + d}, \quad a d-b c =1,\quad  a, b,c,d\in \mathbb{Z}.
\end{equation}
F-theory \cite{Vafa:1996xn} is a non-perturbative approach to  type IIB string theory that implements geometrically several non-trivial constraints of  S-duality  using the  geometry of  elliptic fibrations.  In F-theory, the axio-dilaton field is geometrically  modeled as  the modulus of the fiber of an elliptic fibration $$\varphi:Y\rightarrow B.$$ The base $B$ of the fibration is the space over which  type  IIB string theory is compactified. The size of the elliptic curve is not physical and therefore the elliptic fiber is only defined modulo homothety. It follows that each regular fiber can be expressed as a quotient:  $$\mathscr{E}_\tau:=\mathbb{Z}/(\mathbb{Z}+\tau\mathbb{Z}), \quad \Im (\tau)>0$$
 depending on the modulus $\tau$ living in the complex upper half-plane. 

Since much of our understanding of string theory is based on perturbative   calculations that make sense only for small string coupling $g_s$, it is  useful to understand how  the strongly coupled physics described by  F-theory  flows to a weakly coupled type IIB description. The limit   $g_s\rightarrow 0$ is known as the weak coupling limit of F-theory.  
 In terms of the elliptic fiber, a vanishing string coupling $g_s$ corresponds to an infinite $j$-invariant. This can be seen by considering the Laurent expension of the $j$-invariant in terms of the variable $q:=\exp(2\pi \mathrm{i} \tau)$ parametrizing the  punctured unit disk:
\begin{equation}\label{jinvariant.Laurent}
j(q)=744+\frac{1}{q}+\sum_{n>0} c_n q^n, \quad q:=\exp(2\pi \mathrm{i} \tau).
\end{equation}
In particular, the absolute value of $q$ is related to the inverse of string coupling as:
$$
|q|=\exp(-\frac{2\pi}{ g_s}).
$$
It follows that the weak coupling limit $g_s\rightarrow 0$ is equivalent to  approaching the center of the unit disk ($|q|\rightarrow 0$) and therefore to an infinite  $j$-invariant:
\begin{equation}
 (g_s\rightarrow 0)  \iff   ( j\rightarrow \infty). 
\end{equation}
To make connection with the  perturbative regime of  IIB string theory, one can consider  certain degenerations of the elliptic fibration 
 such that the  string coupling becomes small almost everywhere over  the base $B$. This is called a {\em  weak coupling limit of the elliptic fibration}.
 In the simplest set-up, such degenerations can be expressed in terms of a  family of elliptic fibrations  $\varphi_\epsilon: Y_\epsilon \rightarrow B$ parametrized by  a deformation parameter  $\epsilon$ for which the general fiber of the fibration $\varphi_\epsilon: Y_\epsilon \rightarrow B$ becomes a nodal curve (or more generally a semi-stable curve) as   $\epsilon$ approaches zero.  

\subsection{Sen's limit of Weierstrass models}
Sen \cite{Sen:1996vd,Sen.Orientifold} has proposed a simple realization of  the weak coupling limit   for  an elliptic fibration defined by a  short Weierstrass model
\begin{equation}\label{ShortWeierstrass}
y^2=x^3+ f x + g.
\end{equation}
Geometrically, the main idea of Sen's limit is to express the Weierstrass model as a deformation of a fibration of nodal curves. 
The coefficients  $f$ and $g$ are then polynomials in the deformation parameter $\epsilon$ so that the general fiber is a nodal curve at $\epsilon=0$. This ensures that the $j$-invariant goes to infinity as $\epsilon$ goes to zero. 
Sen's limit is explicitly given in terms of the following expression for $f$ and $g$:
\begin{equation}\label{Sen.original}
\text{Sen's limit}
\  
\begin{cases}
f=-3 h^2+ \epsilon\eta \\
g=-2h^3+ \epsilon h \eta + \epsilon^2 \chi
\end{cases}
\end{equation}
For every fixed value of $\epsilon$,  the variables $\eta$ and $\chi$ ensure that the Weierstrass model is as general as possible. 
The elliptic fibration is then 
\begin{equation}
Y_{(\epsilon)}: y^2=(x + h)^2 (x - 2 h)+\epsilon (\eta x + h \eta )+ \epsilon^2 \chi.
\end{equation}
At $\epsilon=0$, we recognize a fibration of nodal curves:
$$
 Y_{(0)}: \quad y^2=(x+h)^2(x-2h).
$$
 Since a nodal curve has an infinite $j$-invariant, this ensures that the string coupling vanishes over a generic point of  the base as $\epsilon$ approaches zero. 
We can get to the same conclusion by computing the leading terms  of the Laurent expansion of the $j$-invariant as a function of  $\epsilon$: 
\begin{equation} \Delta=-9 \epsilon^2 h^2(\eta^2 +12 h\chi) +O(\epsilon^3), \quad 
 j=  1728 \frac{12 h^4}{\epsilon^2(\eta^2+12 h\chi)}+\sum_{k\geq -1} u_k(h,\eta,\chi) \epsilon^k.\end{equation}
In particular, the $j$-invariant has a pole of order two at $\epsilon=0$.
\begin{remark}
In Sen's limit, $g_s$ goes to zero nearly everywhere.  More precisely, away from $h=0$. Over $h=0$, we have to be more careful as the leading order of $j$ vanishes for non-zero values of $\epsilon$. By first imposing $h=0$ and then taking the limit, one can show that $j=1728$ over a generic point of $h=0$ in the limit $\epsilon\rightarrow 0$.
\end{remark}
\subsection{The orientifold interpretation of Sen's limit}
The monodromy of the axio-dilaton field around $h=0$ due to the behavior of the $j$-invariant $j\sim h^4/\epsilon^2 (\eta^2+12 h\chi)$ indicates that $h=0$ is the location of an O7-plane. 
 We recall that $h$ is a section of $\mathscr{L}^2$ and the discriminant is a section of $\mathscr{L}^{12}$. 
Since $h$ is a section of an even line bundle, it can describe the branch locus of a double cover of the base. 
Explicitly, the double cover $\rho:X\rightarrow B$ of the base $B$ branched along the divisor $\underline{O}\subset B: h=0$ is given by the canonical equation of a double cover:
\begin{equation}
X:\quad \xi^2=h,
\end{equation}
which is automatically Calabi-Yau $n$-fold if the elliptic fibration $\varphi:Y\rightarrow B$ we started with is also Calabi-Yau $(n+1)$-fold.
The weak coupling limit, described as a geometrical construction, does not require the Calabi-Yau condition and can be defined for an elliptic fibration over a base of arbitrary dimension. 
The branch divisor $\underline{O}$ corresponds to the orientifold locus $O:\xi=0$ in the double cover $X$.

The leading term in the discriminant pulls-back to the  double cover $X$ as follows 
\begin{equation}
\rho^\star\Delta=-9\epsilon^2 \xi^4 ( \eta^2+12 \xi^2 \chi)+ O(\epsilon^3).
\end{equation}
The corresponding $j$-invariant is 
\begin{equation}
j  \propto 1728 \frac{12\xi^{8}}{\epsilon^2  ( \eta^2+12 \xi^2 \chi)}
\end{equation}
This is physically described as an orientifold at $O:\xi=0$ and a D7 Whitney-brane at $D_w:\eta^2+12 \xi^2 \chi=0$. 
The divisor $D_w$ has the singularity of a Whitney umbrella:  A double line in codimension-2 that enhances to a locus of  pinch points in codimension-3.

\subsection{Sen's limit for a Weierstrass model in Tate form}

When an elliptic fibration is  given in the Weierstrass equation in Tate form :
\begin{equation}\label{gen.Weierstrass}
y^2+ a_1 xy + a_3 y =x^3+ a_2 x^2 + a_4 x +a_6,
\end{equation}
 Sen's limit can still be defined  by remembering that the previous equation can be put into the short Weierstrass form \eqref{ShortWeierstrass}
$$y^2=x^3-\frac{c_4}{48} x -\frac{c_6}{864},$$
 with $c_4$ and $c_6$ defined as in equation \eqref{eq.formulaire}. 
 Sen's limit as expressed in equation \eqref{Sen.original} corresponds for a   Weierstrass equation in Tate form\eqref{gen.Weierstrass} to the requirement:   
\begin{equation}
\text{Sen's limit} 
\begin{cases}
b_4\rightarrow\epsilon \eta \\
 b_6\rightarrow \epsilon^2 \chi\;.
\end{cases}
 \label{sen}
\end{equation} 
The branch divisor  is given by $\underline{O}: b_2=0$ in the base $B$. It follows that the orientifold is defined by the double cover: 
\begin{equation}
X: \xi^2=b_2, \quad b_2:=a_1^2+4 a_2.
\end{equation}
At leading order in the discriminant locus: 
\begin{equation}
\Delta= \epsilon^2 b_2^2 b_8+O(\epsilon^3).
\end{equation}
It is composed of a factor $b_2^2$ which is the contribution from the branch locus and a factor $b_8:=b_2 b_6-b_4^2$. 
When pulled-back to the double cover we have 
\begin{equation}
\rho^*\Delta= \epsilon^2 \xi^4 (\xi^2 b_6 - b_4^2)+O(\epsilon^3).
\end{equation}
The factor $\xi^4$ is the contribution from the orientifold and the second factor gives a Whitney brane $\xi^2 b_6 - b_4^2=0$.

Since  $b_4=a_1a_3+2a_4$ and $b_6=a_3^2+4a_6$, it is  important to realize that  the limit \eqref{sen} can be implemented in many different ways if we start from Tate's general form of the Weierstrass equation. In the F-theory literature, the realization which is normally used is the ansatz of Donagi and Wijnholt \cite{Donagi:2009ra}. 
We will analyze it in some details in the next section.

\section{The Donagi-Wijnholt ansatz}\label{DWAnsatz}
Donagi and Wijnholt \cite{Donagi:2009ra} have proposed the following realization of  Sen's limit \eqref{sen} valid for elliptic fibrations given by the Tate form of a Weierstrass model:
\begin{equation}\label{sen.Donagi}
\text{Donagi-Wijnholt}
\begin{cases}
a_3\rightarrow \epsilon a_3, \\
a_4\rightarrow \epsilon a_4, \\
a_6\rightarrow \epsilon^2 a_6.
\end{cases}
\end{equation}
Geometrically, this limit  is a degeneration of a Weierstrass model to  a fibration of nodal curves when $\epsilon=0$. 
In the limit $\epsilon=0$, we have $a_3=a_4=a_6=0$, which gives the nodal curve:
\begin{equation}
(y+\frac{1}{2}a_1 x)^2=x^2 ( x+\frac{1}{4} b_2 )\,. 
\end{equation}
The nodal curve specializes to a cusp over the divisor in the base $\underline{O}: b_2=0$. 
Since $b_2=0$ is a section of a line bundle  $\mathscr{L}^2$ , we can define a double cover $\rho: X\rightarrow B$  branched at $h=0$:
\begin{equation}
X:\quad \xi^2=a_1^2+4a_2,
\end{equation}
which is  used to define an orientifold theory in type IIB, understood as the weak coupling limit of F-theory model given by the Weierstrass equation with coefficients $(a_1, a_2, a_3, a_4, a_6)$. 
The double cover $X$ is nonsingular  as long as $a_2=0$ defines a nonsingular divisor in the base. 
As we will see in the following subsections, $X$ is usually singular when $Y$ admits a  gauge group  implemented by  ansatz coming from Tate's algorithm except in the case of  symplectic gauge groups.

\subsection{Gauge groups and singularities in the weak coupling limit}

In F-theory, non-Abelian gauge groups  appear when the elliptic fibration admits reducible singular fibers over a component of the discriminant locus \cite{Morrison:1996na, Bershadsky:1996nh}.  Such singular fibers  located over a codimension-one  locus in the base of an elliptic fibration were classified by Kodaira and N\'eron and admit dual graphs that are  extended ADE Dynkin diagrams \cite{Kodaira,Neron}. Non-simply laced gauge groups can also be obtained 
 by taking into account the 
monodromy action by
an outer automorphism on the nodes of the dual graph of the singular fiber   \cite{Bershadsky:1996nh}. 
When working with a Weierstrass model, non-Abelian gauge groups can occur only when the Weierstrass model becomes singular, as a smooth Weierstrass  model admits only irreducible singular fibers (regular elliptic curves,  nodal curves and cusps).
To implement a certain non-Abelian gauge group over a  divisor $\sigma=0$ of the base, one can use an ansatz inspired directly from  Tate's algorithm  \cite{Bershadsky:1996nh}. These ans\"atze are now familiarly called {\em Tate forms}.
The original list of Tate forms in \cite{Bershadsky:1996nh} was corrected by Grassi and Morrison \cite{GM1}. A Weierstrass model admitting a certain gauge group is not necessarily realized by one of the Tate forms. A careful analysis was done recently to see when it is possible to achieve these forms and more general ans\"atze were presented when it was not possible to do so \cite{Katz:2011qp}.   
 We will refer to the classification in table 2 of \cite{Katz:2011qp} throughout the paper. It is reproduced in table \ref{Table.TateForm}. The ansatz  requests  that each of the  coefficients $(a_1,a_2,a_3,a_4,a_6)$  of the Weierstrass equation  vanishes with a certain multiplicity over  $\sigma=0$ as stipulated by Tate's algorithm. 
The condition for non-simply laced gauge groups are conditions on factorizations of a quadratic or cubic equation defined from the coefficients $a_k$. 
Following the notation familiar from Tate's algorithm, we denote  
\begin{equation}
a_i=a_{i,m_i} \sigma^{m_i}, 
\end{equation}
where $m_i$ denotes the multiplicity of the divisor $\sigma=0$ over the subvariety defined by $a_i=0$. We  assume that $a_{i,m_i}$ is nonzero for a generic point of $\sigma=0$. When such a  gauge group is implemented in this way, we can  realize the Donagi-Wijnholt ansatz as  conditions on $(a_{3,m_3}, a_{4,m_4}, a_{6,m_6})$: 
\begin{equation}\label{DWLimitNonAbelian}
\begin{cases}
a_{3, m_3}\rightarrow \epsilon a_{3,m_3}, \\
a_{4, m_4}\rightarrow \epsilon a_{4,m_2}, \\
a_{6, m_4}\rightarrow \epsilon^2 a_{6,m_6}.
\end{cases} 
\end{equation}
This leads to an orientifold theory defined by  the following  double cover at weak coupling:
\begin{equation}
X:\quad \xi^2=h, \quad \text{where}\quad  h:=\sigma^{2m_1}a_{1,m_1}^2+4 \sigma^{m_2} a_{2,m_2}.
\end{equation}
We see immediately that this double cover is singular whenever $m_2>0$.  This implies that it will be singular for all gauge groups realized through Tate forms with the exception of symplectic gauge groups $\Sp(\lfloor \frac{k}{2} \rfloor)$ obtained from the Tate form for a fiber of type    $I^{ns}_{k}$.  When the double cover is singular, we  have to determine if it admits a crepant resolution compatible with the $\mathbb{Z}_2$ involution of the double cover.  

\begin{table}[bht]
\begin{center}
\begin{tabular}{|c|c|c|c|c|c|c|c|}
\hline
type  &  group  & \quad $ a_1$\quad  &
\quad $a_2$\quad  & \quad $a_3$ \quad & \quad $ a_4 $ \quad& \quad $ a_6$ \quad & $\Delta$ \\
\hline \hline $I_0 $  &  ---  & $ 0 $  & $ 0 $  & $ 0 $  & $ 0 $  &
$ 0$  & $0$ \\
 $I_1 $  &  ---  & $0 $  & $ 0 $  & $ 1 $  & $ 1 $  &
$ 1 $  & $1$ \\ 
$I_2 $  & $SU(2)$  & $ 0 $  & $ 0 $  & $ 1 $  & $ 1
$  & $2$  & $
2 $ \\ $I_{3}^{ns} $  &  $Sp(1)$  & $0$  & $0$  & $2$  & $2$  & $3$  & $3$ \\
$I_{3}^{s}$  & \small $\SU(3)$  & $0$  & $1$  & $1$  & $2$  & $3$  & $3$ \\
$I_{2k}^{ns}$  & $ Sp(k)$  & $0$  & $0$  & $k$  & $k$  & $2k$  & $2k$ \\
$I_{2k}^{s}$  & $SU(2k)$  & $0$  & $1$  & $k$  & $k$  & $2k$  & $2k$ \\
$I_{2k+1}^{ns}$  & $Sp(k)$  &  $0$  & $0$  & {\footnotesize $k+1$}  & {\footnotesize $k+1$}   &
{\footnotesize $2k+1$}   & {\footnotesize $2k+1$} 
\\ $I_{2k+1}^s$  & \footnotesize $SU(2k+1)$  & $0$  & $1$  & $k$  & {\footnotesize $k+1$}  &{\footnotesize $2k+1$}  & {\footnotesize $2k+1$} 
\\ $II$  &  ---  & $1$  & $1$  & $1$  & $1$  & $1$  & $2$ \\ $III$  & $SU(2)$  & $1$
 & $1$  & $1$  & $1$  & $2$  & $3$ \\
 $IV^{ns} $  & $Sp(1)$ & $1$  & $1$  & $1$
 & $2$  & $2$  & $4$ \\ $IV^{s}$  & $SU(3)$  & $1$  & $1$  & $1$  & $2$  & $3$  & $4$
\\ $I_0^{*\,ns} $  & $G_2$  & $1$  & $1$  & $2$  & $2$  & $3$  & $6$ \\
$I_0^{*\,ss}$  & $SO(7)$  & $1$  & $1$  & $2$  & $2$  & $4$  & $6$
\\ $I_0^{*\,s} $  & $SO(8)^*$  & $1$  & $1$  & $2$  & $2$  & $4$  &
$6$ \\ $I_{1}^{*\,ns}$
 & $SO(9)$  & $1$  & $1$  & $2$  & $3$  & $4$  & $7$ \\ $I_{1}^{*\,s}$  & $SO(10) $
 & $1$  & $1$  & $2$  & $3$  & $5$  & $7$ 
\\ $I_{2}^{*\,ns}$  & $SO(11)$  & $1$  & $1$
 & $3$  & $3$  & $5$  & $8$
 \\ $I_{2}^{*\,s}$  & $SO(12)^*$  & $1$  & $1$  & $3$
 & $3$  & $5$ & $8$\\
$I_{2k-3}^{*\,ns}$  &\footnotesize $SO(4k+1)$  & $1$  & $1$  & $k$  & {\footnotesize $k+1$} 
 & $2k$  & \footnotesize $2k+3$ \\ 
$I_{2k-3}^{*\,s}$  &\footnotesize $SO(4k+2)$  & $1$  & $1$  & $k$  & {\footnotesize $k+1$} 
 & {\footnotesize $2k+1$}  &\footnotesize $2k+3$ \\
 $I_{2k-2}^{*\,ns}$  & \footnotesize $SO(4k+3)$  & $1$  & $1$  & \footnotesize $k+1$
 & \footnotesize $k+1$  &\footnotesize $2k+1$  & \footnotesize $2k+4$  \\ 
$I_{2k-2}^{*\,s}$  & \footnotesize $SO(4k+4)^*$  & $1$  & $1$
 & \footnotesize $k+1$   & \footnotesize $k+1$   & \footnotesize $2k+1$
 & \footnotesize $2k+4$  \\
 $IV^{*\,ns}$  & $F_4 $  & $1$  & $2$  & $2$  & $3$  & $4$
 & $8$\\ $IV^{*\,s} $  & $E_6$  & $1$  & $2$  & $2$  & $3$  & $5$  &  $8$\\
$III^{*} $  & $E_7$  & $1$  & $2$  & $3$  & $3$  & $5$  &  $9$\\
$II^{*} $
 & $E_8\,$  & $1$  & $2$  & $3$  & $4$  & $5$  &  $10$ \\
\footnotesize non-min  &  ---  & $ 1$  & $2$  & $3$  & $4$  & $6$  & $12$ \\
\hline
\end{tabular}
\end{center}
 \caption{ {\bf Tate forms in F-theory}.   
The superscript (s/ns/ss) stands for
(split/non-split/semi-split), meaning that  (there is/there is not/ there is  a partial)  monodromy action by
an outer automorphism on the vanishing cycles along the singular  locus. }\label{Table.TateForm}
\end{table}

\subsection{Singular double covers from the  Donagi-Wijnholt limit}
We would like to classify the  types of singularities that  occur when the weak coupling limit is reached through the  Donagi-Wijnholt realization of Sen's limit in presence of gauge groups implemented by Tate forms. First we note that the different Tate forms for singular fibers can be organized into four groups characterized by the vanishing multiplicity $(m_1,m_2)$ of the coefficients $a_1=\sigma^{m_1} a_{1,m_1}$ and $a_2=\sigma^{m_2} a_{2,m_2}$:
\begin{itemize}
\item $(m_1,m_2)=(0,0)$ for symplectic groups realized by fibers $I_k^{ns}$ and  $\SU(2)$ realized by a fiber $I_2$.\footnote{We may realize SU(2) with slightly less generic $I_2$ fibers, i.e. with $a_2$ having order of vanishing $1$ along $\sigma=0$. Its weakly coupled physics is very different from the more generic realization and it belongs to the category $(m_1,m_2)=(0,1)$ (see \cite{Collinucci:2010gz} about this distinction).} 
They lead to smooth double covers at weak coupling. 
\item $(m_1,m_2)=(0,1)$ for unitary groups realized by fibers $I^s_k$. They  lead to conifold singularities in the double cover.
\item $(m_1,m_2)=(1,1)$ for  fibers  $I^* _k$ (orthogonal groups $SO(r)$ and  the exceptional group $G_2$) and fibers $III$ and $IV$  (leading to $\Sp(1)$  and $\SU(3)$).  They lead to quadric cone singularities.
\item $(m_1,m_2)=(1,2)$ for exceptional groups $F_4$, $E_6$, $E_7$ and $E_8$. They lead  to Whitney umbrella singularities. 
\end{itemize}
 This is summarized  in table \ref{table.DW}. 
The case $(m_1,m_2)=(0,1)$ is special in the sense that in contrast to the quadric cone singularity and the Whitney umbrella, the conifold singularities do not admit crepant resolutions compatible with the $\mathbb{Z}_2$ involution. This is a serious problem for phenomenological model building based on SU(5) Grand Unified Theories. We will explain how to resolve that problem in section \ref{sppLimit} and get the right GUT group on a D7-stack. In section \ref{resCY3spp}, we will instead address several other features (and issues) of the GUT theories so obtained.  

\begin{table}[htc]
{\footnotesize
\begin{tabular}{|c|lcl|c|}
\hline
$(m_1,m_2)$ & \multicolumn{3}{c|}{Double cover} & 
{
\begin{tabular}{l}
group 
over $\sigma=0$
\end{tabular}
} \\
\hline
$(0,0)$ & $\xi^2 = a_{1,0}^2 + 4 a_{2,0}$ &:& smooth & Symplectic \\
\hline
$(0,1)$ &  $\xi^2 = a_{1,0}^2 + 4 \sigma a_{2,1} $ &:& conifold  & Unitary\\
\hline
$(1,1)$ &  $\xi^2 = \sigma r_{1,1}$ &:& quadric cone & Orthogonal \\
\hline
$(1,2)$ &   $\xi^2 = \sigma^2 r_{1,2}$ &:& Whitney umbrella  & $F_4$, $E_6,E_7$ and $E_8$\\
\hline
\end{tabular}}
\caption{ \footnotesize Singular orientifolds for the weak coupling limit $(a_3, a_4, a_6)\rightarrow (\epsilon a_3, \epsilon a_4, \epsilon^2 a_6)$.
This weak coupling limit  is the ansatz used by Donagi-Wijnholt \cite{Donagi:2009ra}.
 In the first column $(m_1,m_2)$ are such that $a_1=\sigma^{m_1} a_{1,m_1}$ and $a_2=\sigma^{m_2} a_{2,m_2}$. 
In the second column    $r_{m_1,m_2}:=\sigma^{2m_1-m_2} a_{1,1}^2 + 4 a_{2,m_2}$.
}\label{table.DW}
\end{table}

\subsubsection{ $SO(k)$ , $G_2$, $\Sp(1)$ and $\SU(3)$ and quadric cone singularities}

 When $(m_1, m_2)=(1,1)$,  we have the double cover  
\begin{equation}
X:\quad \xi^2= \sigma r_{1,1}, \quad \text{where}\quad  r_{1,1}=\sigma a_{1,1}^2+a_{2,1}.
\end{equation}
It has the singularity of a {\em quadric cone}. The singularity is  the codimension-2 locus $\xi=\sigma=r=0$. The double cover $X$  admits a crepant resolution which is also a double cover. 
This geometry characterizes the orthogonal gauge group obtained in F-theory by Tate form. The exceptional gauge group $G_2$ is obtained from a non-split fiber $I^*_0$ and therefore also leads to such a singular double cover.  For small rank groups,  the Tate form for  $\Sp(1)$ (with a fiber $III$ or $IV^{ns}$) and $\SU(3)$ with a fiber $IV^s$ all have $(m_1, m_2)=(1,1)$.

\subsubsection{Exceptional groups $F_4,E_6, E_7, E_8$ and Whitney umbrella } 
The Tate form for exceptional groups $F_4, E_6, E_7, E_8$ have  $(m_1,m_2)=(1,2)$ and at weak coupling  using the Donagi-Wijnholt ansatz, the orientifold is defined through a double cover with the singularities of a Whitney umbrella: 
\begin{equation}
X:\quad \xi^2= \sigma^2 r_{1,2}, \quad \text{where}\quad  r=a_{1,1}^2+a_{2,2}.
\end{equation}
As the singularity can be resolved by blowing-up the codimension-2 locus $\sigma=\xi=0$ of multiplicity 2, the double cover admits a crepant resolution compatible 
with the $\mathbb{Z}_2$ involution.

 \subsubsection{Unitary groups and conifold singularities.}
 Unitary groups require fibers of type $I^s_k$. In the usual ansatz from Sen's algorithm, they are implemented with the conditions 
 $(m_1,m_2)=(0,1)$ which implies that the double cover obtained at weak coupling, using the Donagi-Wijnholt  ansatz, is:
\begin{equation}\label{ConifoldSingularity}
X:\quad \xi^2= a_{1,0}^2 + 4 \sigma a_{2,1}\,. 
\end{equation}
This admits conifold singularities in codimension-3 at  $\xi=a_{1,0}=\sigma=a_{2,1}=0$. Such singularities admit  crepant resolutions. 
However, these crepant resolutions are not compatible with the double cover: There are in fact two small resolutions which are exchanged by the orientifold action (see appendix \ref{AppConifold}). In contrast, the standard blow-up of the conifold is non-crepant.

\subsection{Solving the conifold problem for the Tate form of unitary gauge groups}\label{SolvingConifold}

We can solve the conifold problem appearing in the Donagi-Wijnholt ansatz of unitary gauge groups  in at least two different ways, each leading to a very different physical  picture at weak coupling. 
This is done by slightly modifying the original Donagi-Wijnholt ansatz.
\subsubsection{Replacing the conifolds by  quadric cones}
The conifold singularities can be removed from the weak coupling limit of F-theory with $SU(n)$ gauge groups by  supplementing the Donagi-Wijnholt ansatz with the additional condition  $a_1\rightarrow \epsilon a_1$. This gives the following  limit:
\begin{equation}\label{OrthogonalAnsatz}
\begin{cases}
a_1\    \rightarrow \epsilon a_1,\\
a_{3, m_3}  \rightarrow \epsilon a_{3,m_3}, \\
a_{4, m_4}\rightarrow \epsilon a_{4,m_2}, \\
a_{6, m_4}\rightarrow \epsilon^2 a_{6,m_6}\,.
\end{cases}
\end{equation}
This limit is in fact equivalent to the Donagi-Wijnholt original one, \eqref{DWLimitNonAbelian}, for most of Kodaira singularities\footnote{We are grateful to Andr\'es Collinucci for having pointed out this equivalence to us.}. In fact the equivalence breaks down when and only when $\sigma$ divides $a_2$ but not $a_1$. By inspecting table \ref{Table.TateForm}, we find that this circumstance is realized for the $SU(n)$ tower only, which is our focus here. 
The discriminant locus of the Weierstrass model  becomes 
\begin{equation}
\Delta \propto h^2 (h \sigma^s b_{6,s} -\sigma^{2m_4} a_{4,m_4}^2)\,\epsilon^2,
\end{equation}
with   $h= 4 \sigma^{m_2} a_{2,m_2}$. 
For symplectic gauge groups this gives a smooth double cover as $m_2=0$. For Tate form with $(m_1,m_2)=(1,1)$ or $(1,2)$ we recover the same geometry at weak coupling as with the original Donagi-Wijnholt ansatz. 
However, for  unitary gauge groups implemented using Tate forms, the conifold singularities are replaced by the quadric cone singularities as the double cover is now:
\begin{equation}\label{QuadricConeSing}
X:\quad \xi^2= 4 \sigma a_{2,1}. 
\end{equation}
Such singularities admit a crepant resolution compatible with the double cover. 
The orientifold locus splits into two components, namely $\xi=\sigma=0$ and $\xi=a_{2,1}=0$. One of these components happens to be the divisor on which the group is defined. It follows that we expect orthogonal groups at weak coupling. 

At leading order for a group $\SU(2n)$, the discriminant is 
\begin{equation}
\Delta \propto\epsilon^2  a_{2,1}^2  \sigma^{2n+2} (4 \sigma a_{2,1} b_{6,2n} - a_{4,n}^2),
\end{equation}
and for $\SU(2n+1)$
\begin{equation}
\Delta \propto \epsilon^2  a_{2,1}^2 \sigma^{2n+3} (4 a_{2,1}   b_{6,2n} -\sigma a_{4,n}^2).
\end{equation}
In this limit, the gauge group $\SU(k)$ seen at strong coupling becomes an orthogonal group $\SO(2k)$ at weak coupling as $\sigma=0$ is a component of the branch locus of the double cover. 
In the case of $\SU(k)$, the power $\sigma^{k+2}$ in the discriminant can be understood as composed of a factor $\sigma^2$ contributing to the orientifold and the leftover  $\sigma^{k}$ corresponding to $k$ bibranes\footnote{By a {\em bibrane} we mean a  brane-image-brane pair.} on top of the component $\sigma=0$ of the orientifold. This leads to a gauge group $\SO(2k)$. 
For $\SU(2n)$ there is also a singular brane 
\begin{equation}
D:4 \sigma a_{2,1} b_{6,2n} - a_{4,n}^2=0,
\end{equation}
 which becomes 
\begin{equation}
D:4 a_{2,1}   b_{6,2n} -\sigma a_{4,n}^2=0,
\end{equation}
 for $\SU(2n+1)$.

\subsubsection{Replacing the conifold by a suspended pinch point}\label{sppLimit}
For many applications, one would like to retrieve a unitary gauge group at weak coupling.  Preserving the unitary gauge group in presence of a $\mathbb{Z}_2$ orientifold  requires the  presence at  weak coupling of a stack of branes  not coinciding with its image stack under the orientifold involution. This happens in the conifold geometry because as  $\sigma=0$, we get  two divisors in the double cover and they are image of each other under the involution, namely:
\begin{equation}
\xi\pm a_1=\sigma=0.
\end{equation}
We can keep that property while modifying the singularity so that we can have a crepant resolution  compatible with the double cover. This would be the case of a double cover of the type 
\begin{equation}
X:\quad \xi^2=u^2+4\sigma^s v,\quad s=2 \quad \text{or}\quad  s=3.
\end{equation}
The simplest choice is $s=2$, which describes a suspended pinch point also known as a suspended Whitney umbrella. 
The suspended pinch point can be obtained by using the following modification of the Donagi-Wijnholt ansatz:
\begin{equation}\label{SppLimitNonAbelian}
\begin{cases}
a_{2,1}\rightarrow \epsilon  a_{2,1} +  \sigma a_{2,2},\\
a_{3, m_3}  \rightarrow \epsilon a_{3,m_3}, \\
a_{4, m_4}\rightarrow \epsilon a_{4,m_2}, \\
a_{6, m_4}\rightarrow \epsilon^2 a_{6,m_6}.
\end{cases}
\end{equation}
This leads to a double cover with the singularities of  a suspended pinch point:
\begin{equation}\label{SppSing}
X:\xi^2 = a_1^2 + 4\sigma^2 a_{2,2}.
\end{equation}
This double cover is singular along the codimension-two locus $\xi=a_1=\sigma=0$. 
A viable resolution in this case does exist: Indeed there are three small resolutions of the suspended pinch point, two of them are exchanged by the orientifold involution, and the third one, absent for the conifold, is orientifold invariant (see appendix \ref{AppSpp}). It turns out that the latter is equivalent to the standard ``large'' blow-up of the suspended pinch point singularity.

\subsection{ Generalities on weak coupling limits for Tate forms}\label{GeneralitiesWCL}

In F-theory, the non-Abelian part of the gauge group is completely controlled by the Kodaira type of the singular fiber over components of the discriminant locus and the monodromy around them. 
As we take the weak coupling limit, the discriminant locus  can be deformed and provides a very different spectrum of branes than what is seen in the full F-theory regime. 
When the weak coupling limit is an orientifold theory, a stack of branes gives a gauge group that can be symplectic, orthogonal or unitary depending on the behavior of the stack with respect to the orientifold symmetry. There are 3 cases to consider\footnote{We will not discuss the presence of  $U(1)$ factors.}:
\begin{enumerate}
\item Symplectic groups: the stack of D7 branes is supported on a divisor  invariant under the involution but not  pointwise invariant.
\item Orthogonal groups: the stack of D7 branes is supported on a divisor pointwise  invariant under the orientifold involution. 
\item Unitary groups:    the stack of D7 branes is supported on a divisor  which admits a distinct orientifold image.  
\end{enumerate}

Assuming that at weak coupling we have a stack of $r$ branes over the divisor $\underline{\Sigma}:\sigma=0$ in the base, the discriminant locus is of the following form at leading order in the deformation parameter of the weak coupling limit:
\begin{equation}
\Delta\propto  h^2 \sigma^r (\dots),
\end{equation}
where $h=0$ is the branch locus of the double cover that defines the orientifold theory:
\begin{equation}
X:\quad \xi^2=h, \quad f=-3 h^2 +\dots
\end{equation}
Here $h$ could contain factors of $\sigma$ as well. 
We denote by $h_\sigma$ the restriction of $h$ to the divisor  $\sigma=0$:
\begin{equation}
h_\sigma:=h\Big|_{\sigma=0}.
\end{equation}
The gauge group associated with the stack depends on $r$ and on  the factorization properties of $\xi^2=h_\sigma$. This is reviewed in table \ref{table.stackPB}.
\begin{table}[htb]
\begin{center}
\begin{tabular}{|c|c|}
\hline 
Property of $h_\sigma$ & Gauge group \\
\hline
$h_\sigma$  {is identically zero} & $\SO(2r)$ \\
\hline
$h_\sigma$   {is  a perfect square} &    $\SU(r)$ \\
\hline
$h_\sigma$   {is  not a  perfect square} & $\Sp\Big(\big\lfloor \frac{r}{2}\big\rfloor\Big)$\\
\hline
\end{tabular}
\end{center}
\caption{ The discriminant locus at weak coupling is $\Delta\propto  h^2 \sigma^r (\cdots)$ and  $h_\sigma:=h\Big|_{\sigma=0}.$}\label{table.stackPB}
\end{table}

Supposing that 
$\Delta\propto  h^2 \sigma^r (\cdots)$, we get a unitary, orthogonal or symplectic gauge group  by the following ans\"atze for $h$:
\begin{subequations}\label{wclr}
\begin{align}
 h:= u^2 + \sigma^k  v & \quad \Longrightarrow\quad  \SU(r), \\
h := \sigma^k v & \quad \Longrightarrow\quad  \SO(2r),\\
 h\ \text{generic} &  \quad \Longrightarrow \quad \Sp\Big(\big\lfloor \tfrac{r}{2}\big\rfloor\Big),
\end{align}
\end{subequations}
where $\lfloor x \rfloor$ denotes the integral part of $x$.

\subsection*{Symplectic case: $\xi^2=h$} 
The double cover $\xi^2=h$ with $h$ general leads to a  symplectic gauge group and it is smooth. This is the generic case. 

\subsection*{Orthogonal case: $\xi^2=\sigma^k v$}
The double cover $\xi^2=\sigma^k v$ leads to an orthogonal gauge group over $\sigma=0$. The rank of the orthogonal group depends on the multiplicity of the discriminant locus. The double cover is singular whenever $k>0$. The singularities are in codimension-1  if $k>1$  and codimension-2 if $k=1$. They admit an admissible crepant resolution for $0 \leq k\leq 3$. This is reviewed in appendix \ref{AdmissibleResolutions}.

\subsection*{Unitary case: $\xi^2=u^2+\sigma^k v$}
The double cover  $X: \xi^2=u^2+\sigma^k v$ will give a unitary gauge group  for the stack over $\sigma=0$ since the divisor  $\sigma=0$ of the base pulls-back to two distinct divisors  in the double cover, namely $D_{\sigma_\pm}: \sigma=\xi\pm u=0$. The group at weak coupling will be $\SU(r)$ if the leading term of the  discriminant is of the type $\Delta\propto h^2 \sigma^r (\cdots)$ where $h=u^2+\sigma^k v$.  The question is whether $X$ admits an admissible  crepant resolution. 
If $k>1$, the singularity are in codimension-2. If $k=1$, the singularity jumps to codimention-3 and corresponds to conifold singularities. Such conifold singularities  do not admit a crepant resolution compatible with the involution of the double cover.   If $k=2$ or $k=3$, the double cover admits an admissible crepant resolution. 

\subsection{Brane spectrum at weak coupling  for the DW and the orthogonal limit}\label{BraneContentDW}
For a weak coupling limit compatible with Sen's limit, the  discriminant locus at leading order is 
\begin{equation}\label{LeadingDiscriminantDW}
\Delta\propto\epsilon^2 h^2 b_8= \epsilon^2 h^2 (h\sigma^s  b_{6, s}-\sigma^{2q} b_{4,q}^2)+O(\epsilon^3).
\end{equation}
Taking $r=min(s, 2q)$, we have the following brane content: 
$$
\begin{cases}
 \text{an orientifold at $h=0$,}\\
\text{a stack of $r$ branes over $\sigma=0$,}\\
\text{a singular brane $ (h\sigma^{s-r}  b_{6, s}-\sigma^{2q-r} b_{4,q}^2)=0$.}
\end{cases}
$$

\begin{table}[htb]
\begin{tabular}{|c|c|c|r c|}
\hline
Type & $j$ &F theory  &  \multicolumn{2}{c|}{DW limit / Quadric cone } \\
 \hline
$I_2$ & $\infty$ &  $\SU(2)$  & \multicolumn{2}{c|}{$\SU(2)$ } \\
\hline
$I_3^{ns}$ & $\infty$ &  $\Sp(1)$ & \multicolumn{2}{c|}{ $\Sp(1)$}   \\
\hline
$ I^s_3$ & $\infty$ & $\SU(3)$  &  \multicolumn{1}{|c|}{\small $\SU(3)$} & \multicolumn{1}{|c|}{\small $\SO(6)$}
\\
\hline
$I^{ns}_{2n}$ &$\infty$ &  $\Sp(n)$ &  \multicolumn{2}{c|}{$\Sp(n)$ } \\
\hline
$I^s_{2n}$ & $\infty$ & $\SU(2n)$  &  
 \multicolumn{1}{|c|}{\small $\SU(2n)$} & \multicolumn{1}{|c|}{\small $\SO(4n)$}\\
\hline
$I^{ns}_{2n+1}$ & $\infty$ & $\Sp(n)$ & \multicolumn{2}{c|}{$\Sp(n)$ }\\
\hline
$I^{s}_{2n+1}$ & $\infty$ &\small $\SU(2n+1)$  & 
 \multicolumn{1}{|c|}{\small $\SU(2n+1)$} & \multicolumn{1}{|c|}{\small $\SO(4n+2)$}\\
\hline
$II$ & $0$ & $-$  & \multicolumn{2}{c|}{$\SO(4)$ } \\
\hline
$III$ & \footnotesize $1728$ & $\SU(2)$ & \multicolumn{2}{c|}{ $\SO(4)$ }\\
\hline
$IV^{ns}$ & $0$ & $\Sp(1)$  & \multicolumn{2}{c|}{ $\SO(6)$ } \\
\hline
$IV^s$ & $0$ & $\SU(3)$ & \multicolumn{2}{c|}{  $\SO(6)$ }\\
\hline
$I^{*ns}_0$ & $\infty$ & $G_2$ & \multicolumn{2}{c|}{ $\SO(8)$ }  \\
\hline
$I^{*ss}_0$ & $\infty$ & $\SO(7)$ & \multicolumn{2}{c|}{ $\SO(8)$ 
}\\
\hline
$I^{*s}_0$ & $\infty$ & $\SO(8)$ & \multicolumn{2}{c|}{ $\SO(8)$ } \\
\hline
$I^{*ns}_1$ &$\infty$ & $\SO(9)$ & \multicolumn{2}{c|}{ $\SO(10)$ } \\
\hline
$I^{*s}_1$ & $\infty$ &$\SO(10)$ &\multicolumn{2}{c|}{  $\SO(10)$ }\\
\hline
$I^{*ns}_2$ &$\infty$ & $\SO(11)$ & \multicolumn{2}{c|}{ $\SO(12)$ }\\
\hline
$I^{*s}_2$ & $\infty$ &$\SO(12)$ &  \multicolumn{2}{c|}{ $\SO(12)$}\\
\hline
$I^{*ns}_{2n-3}$ &$\infty$ & \small $\SO(4n+1)$& \multicolumn{2}{c|}{\small $\SO(4n+2)$ } \\
 \hline
$I^{*s}_{2n-3}$ & $\infty$ & \small  $\SO(4n+2)$ & \multicolumn{2}{c|}{\footnotesize$SO(4n+2)$ } \\
\hline
$I^{*ns}_{2n-2}$ &$\infty$ & \small $SO(4n+3)$ & \multicolumn{2}{c|}{\small $\SO(4n+4)$ }   \\
\hline
$I^{*s}_{2n-2}$ & $\infty$ &\small $SO(4n+4)$  & \multicolumn{2}{c|}{\small  $\SO(4n+4)$ } \\
\hline
$IV^{*ns}$ &$0$ & $F_4$ & \multicolumn{2}{c|}{ $\SO(12)$ }  \\
\hline
$IV^{*s}$ & $0$ &$E_6$  &  \multicolumn{2}{c|}{ $\SO(12)$ } \\
\hline
$III^*$ &\footnotesize $1728$ &$E_7$ & \multicolumn{2}{c|}{ $\SO(12)$ } \\
\hline
$II^*$ &$0$ & $E_8$ & \multicolumn{2}{c|}{   $\SO(14)$ } \\
\hline
\end{tabular}
\caption{Groups at weak coupling using the DW ansatz or the orthogonal ansatz.  There is a difference only for fibers $I^s_k$ ($k>2$) which gives $\SU(k)$ in F-theory and also in the DW limit, but $\SO(2k)$ in the orthogonal limit.
\label{table.groups.DW} }
\end{table}

Using the rules explained in  \eqref{wclr}, we compute in table \ref{table.groups.DW} the gauge group in the weak coupling limit for each Weierstrass model with a given singularity over a divisor $\sigma=0$ implemented by the Tate form. We use and compare the Donagi-Wijnholt ansatz \eqref{DWLimitNonAbelian} and the orthogonal ansatz \eqref{OrthogonalAnsatz} to take the weak coupling limit. By definition of weak coupling limit, the $j$-invariant is generically going to infinity over the base so that the string coupling goes to zero. However, the $j$-invariant can still be finite over certain sub-loci of the base. A natural question is whether the coupling is actually small over the divisor $\sigma=0$ on which the gauge group is implemented by the Tate form. 

The appearance of orthogonal gauge groups in the Donagi-Wijnholt weak coupling limit of exceptional singularities (II, III, IV and their duals) deserves some comment. In particular, the gauge groups of the E series fall in this category.  These cases are indeed special with respect to all the others, as the value of the $j$-invariant is finite on the locus $\sigma=0$ where the singularity is implemented in F-theory. This fact makes open strings ending there intrinsically strongly coupled and therefore the presence of the listed gauge symmetries is questionable. We have deduced them first by looking at the order of vanishing of the $b_8$ factor in the leading term of the discriminant \eqref{LeadingDiscriminantDW}; Second by remembering that for these cases $h=\sigma^2(a_{1,1}^2+4a_{2,2})$ and thus $\sigma=0$ is a branch of the O7-plane $h=0$. The same situation actually arises when taking the orthogonal limit of the fibers $I_2$ and $I_3^s$. This is not surprising, because the orthogonal limit makes the stack and the image-stack degenerate onto the O7-plane and the $I_2, I_3^s$  fibers have special enhancements on the orientifold, namely to $III, IV^s$ respectively.

The results we derive in table \ref{table.groups.DW} for the gauge groups at weak coupling do not actually depend on the dimension of the F-theory compactification manifold and are purely based on geometrical facts. Their string interpretation at weak coupling, though, is  puzzling. However, there is one case where  we do have a reliable open string picture to compare those predictions with, and this is when we compactify F-theory on K3. Here, indeed, 7-branes are not intersecting and we have the technology of the so called ``A-B-C'' branes \cite{Johansen:1996am,Gaberdiel:1997ud} at hand to identify the BPS states responsible for the gauge symmetry. In particular, the group E$_8$ is realized \cite{Johansen:1996am} via the bound state $A^7BC^2$, where the group $BC$ has the monodromy of an O7-plane. By higgsing one of the C-branes, we immediately realize the appearance of the $SO(14)$ group in perturbation theory using only the  open fundamental strings ending on the A-branes, possibly winding around the O7-plane $BC$ (this is the result of the perturbative enhancement of the manifest $SU(7)$ group). We can repeat the same reasoning for E$_7$ realized as $A^6BC^2$ and we deduce the perturbative group $SO(12)$ in agreement with what we found here. 

However, the agreement does not seem to exist for E$_6$, realized as $A^5BC^2$, where a perturbative subgroup $SO(10)$ appears, rather then the $SO(12)$ deduced from the rank of the discriminant. This might be due to the fact that the two weak coupling limits to type IIB on $T^2/Z_2$ we are comparing are inequivalent. Moreover, for compactifications of F-theory on higher dimensional manifolds, the interpretation in terms of A-B-C branes is no longer possible in a globally well defined way, due to 7-brane intersections. Yet the computation done for table \ref{table.groups.DW} is still valid, as it does not depend on the dimension. Therefore a comparison analogous to the one above can only be done locally\footnote{See \cite{Bonora:2010bu}, where this kind of local analysis has been used to identify the string-junction states  in the adjoint of non-simply-laced gauge groups, which are not realizable in F-theory on K3 due to the absence of monodromies.} (away from the loci of symmetry enhancement). We hope to come back to these issues in a future work.

\section{The type IIB Calabi-Yau threefold:  Suspended pinch point case}\label{resCY3spp}

In this section we discuss the smooth background where type IIB strings are actually leaving at weak coupling. We address here the case where the weak coupling limit gives us a singular geometry of the suspended pinch point type for the type IIB Calabi-Yau threefold (see section \ref{sppLimit}). We first provide a mathematical description of the resolution procedure and afterwards discuss the physics of the ensuing smooth geometry. 

From now on we restrict our attention to elliptic fibration which are Calabi-Yau
and thus impose $c_1(B )=c_1(\mathscr{L} )$. We will write $c_1$ for $c_1(B)$. We also restrict to 3 the complex dimensions of the base.

\subsection{ Description of the resolution}\label{MathematicalProperties}
 
In order to blow-up the singular Calabi-Yau threefold \eqref{SppSing} along the curve $\sigma=\xi=a_1=0$, we proceed using toric methods (see \cite{Collinucci:2010gz,Collinucci:2012as} for the same methods applied to elliptic fourfolds). We first add to the ambient four-dimensional manifold two homogeneous coordinates, $s$ and $a$, together with two new equations. The singular Calabi-Yau threefold will thus be expressed as the following system of equations
\begin{equation}
X_3:
\begin{cases}
\xi^2&= a^2+s^2\,a_{2,2}\\ 
s&=\sigma \\
 a&=a_1\,,
\end{cases}
\end{equation}
where we have reabsorbed the irrelevant factor of $4$ in $a_{2,2}$. We then introduce yet another homogeneous coordinate, $w$, together with the following projective weight assignment
 \begin{equation}\label{projectWeights}
\begin{array}{cccc}s&a&\xi&w\\ \hline 1&1&1&-1\end{array}\,.
\end{equation}
This will produce an element in the Stanley-Reisner ideal of the ambient six-fold of the form $\xi\,s\,a$: Now these three coordinates cannot simultaneously vanish. The resolved Calabi-Yau threefold will then appear as the following complete intersection in the ambient six-dimensional manifold
\begin{equation}\label{resSPP}
\tilde{X}_3  :
\begin{cases}
\xi^2        &=a^2+s^2  a_{2,2}\\ 
ws  &=\sigma \\ 
w a          &=a_1 
\end{cases}
\end{equation}
This is now a perfectly smooth manifold, still invariant under the orientifold involution\footnote{In this case one can alternatively define the orientifold involution by reversing the sign of $s,a,w$ at the same time. This is clearly gauge-equivalent to sending $\xi\to-\xi$.}. This is where type IIB strings are supposed to live at weak coupling. Let us study some of the properties of this new geometry. 

The stack of D7-branes and its orientifold image are described by the following systems
\begin{equation}
{\rm D7}_\pm:
\begin{cases}
\xi &= \pm1\\ 
s&=0 \\ 
\sigma&=0\\ 
w&=a_1
\end{cases}
\end{equation}
where we have fixed the gauge associated to projective scaling \eqref{projectWeights} by putting $a=1$. While in the singular geometry the D7-stack was intersecting its image in a curve, in the resolved geometry they clearly do not touch each other. They also do not touch the orientifold plane, which is the surface
\begin{equation}
\label{O-plane}
{\rm O7}:
\begin{cases} 
\xi&=0\\ 
a^2&=-a_{2,2}\\ 
w\,a&=a_1 \\ 
w&=\sigma
\end{cases}
\end{equation}
where again we have conveniently fixed the gauge by putting $s=1$. We easily recognize from \eqref{O-plane} a surface wrapping the divisor $\{a_1^2+\sigma^2a_{2,2}=0\}$ of the original base $B_3$, as it should be. On the other hand a new divisor appears, which replaces the former curve of singularities. This is the exceptional divisor
\begin{equation}\label{ExceptionalSpp}
E:
\begin{cases}
w&=0\\ 
a_1&=0 \\ 
\sigma&=0\\ 
\xi^2&=a^2+s^2\,a_{2,2}
\end{cases}
\end{equation}
which has the geometry of an orientifold-invariant. It corresponds to a   $\mathbb{P}^1$ with homogeneous coordinates $a,s$, fibered over the locus $\{\sigma=0\}\cap\{a_1=0\}\subset B_3$. It interpolates between the D7-stack and the O-plane by intersecting both of them respectively in the following \emph{distinct} points of the fiber $\mathbb{P}^1_{s a}$: $(s,a)=(0,1)$ and $(s,a)=(1,p)$ such that $p^2=-a_{2,2}$.

 The base $\tilde{B_3}$ onto which the resolved Calabi-Yau threefold $\tilde{X_3}$ projects is  the original base $B_3$ blown-up along the curve $\{\sigma=0\}\cap\{a_1=0\}$. In other words, $\tilde{X}_3$ can be seen as the double cover of the blown-up base $\tilde{B_3}$ defined by the set of equations
\begin{equation}\label{resolvedBase}
\tilde{B}_3:
\begin{cases}
w\,s&=\sigma \\ 
w a&=a_1\,.\end{cases}
\end{equation}
in an ambient five-dimensional manifold given by adding to $B_3$ the following three homogeneous coordinates with a projective weight assignment
\begin{equation}\label{projectWeightsBase}
\begin{array}{ccc}s&a&w\\ \hline 1&1&-1\end{array}\,.
\end{equation}
Hence $s$ and $a$ cannot vanish at the same time and the exceptional locus is $\mathbb{P}^1_{s a}$ fibered over the curve $\{\sigma=a_1=0\}$.

Before discussing the features of the suspended pinch point from the physics perspective, a comment is in order. The first Chern class of the new base can be expressed in terms of the one of the old base as follows
\begin{equation}\label{nonCrepant}
c_1(\tilde{B}_3)=\phi^*c_1(B_3)-E\,,
\end{equation}
where $\phi:\tilde{B}_3\to B_3$ is the blow-down map and $E$ is the class of $\{w=0\}$. This map is not crepant if the corresponding map between the double covers is. Therefore it may happen that starting from a spin base we end up having a non-spin base after the blow-up. Let us show this circumstance in a concrete example. Take $B_3=\mathbb{P}^3$, which is clearly spin, and call $H$ its hyperplane class. We want to blow up this manifold along the curve $\{x_1=0\}\cap\{a_1=0\}$ where $x_1$ is one of the homogeneous coordinates of $\mathbb{P}^3$ and $a_1$ is a polynomial in $\mathbb{P}^3$ of class $4H$. Hence, the blown-up threefold $\tilde{B}_3$ will be given by the hypersurface of class $4H$
\begin{equation}\label{EqBTilde}
wa=a_1(wx_1,x_2,x_3,x_4)\,,
\end{equation}
in the ambient four-dimensional toric manifold defined by 
\begin{equation}
X_{4}:\begin{array}{cccccc}x_1&x_2&x_3&x_4&a&w\\ \hline 1&1&1&1&4&0\\  1&0&0&0&1&-1\end{array}
\end{equation}
The Stanley-Reisner ideal of this ambient variety is made by the two elements $x_1\,a$ and $x_2\,x_3\,x_4\,w$. We now want to prove that $\tilde{B}_3$ is non-spin. To this end, we have to find at least one 2-cycle on which the first Chern class integrates to an odd number. It turns out that this 2-cycle is not manifest in a generic point of the moduli space of $\tilde{B}_3$. But if we constrain the complex structure moduli in a suitable way we are able to write this 2-cycle as a set of three algebraic equations in the ambient fourfold which automatically satisfy \eqref{EqBTilde} (see \cite{Braun:2011zm,Collinucci:2009uh}, where similar techniques are used for elliptic fourfolds). One possible constraint which works is the following
\begin{equation}
a_1=x_2\,\hat{a}_1+x_3\,\tilde{a}_1\,,
\end{equation}
where $\hat{a}_1,\tilde{a}_1$ are both polynomials of degree $3H$. Now consider the following non-complete intersection 2-cycle
\begin{equation}
C_{(2)}:\left\{\begin{array}{rcl}x_2&=&0\\ x_3&=&0 \\ w&=&0\,.\end{array}\right.
\end{equation}
The integral of the first Chern class of $\tilde{B}_3$ on $C_{(2)}$ is 
\begin{equation}
\int_{C_{(2)}}c_1(\tilde{B}_3)=\int_{X_4}(4H-E)\,E\,H^2\;=\;1\,.
\end{equation}
Here we have used the following intersection numbers of the ambient fourfold
\begin{equation}
H^4=\frac{1}{4}\;,\;H^3E=0\;,\;H^2E^2=-1\;,\;HE^3=-5\;,\;E^4=-21\,.
\end{equation}

\subsection{Physical properties}\label{PhysicalProperties}

The hope is now to use the resolved Calabi-Yau geometry \eqref{resSPP} as the target space for weakly coupled type IIB strings. Their perturbative and non-perturbative dynamics should effectively reproduce the strongly coupled physics of the corresponding F-theory configuration at each codimension in the base: Gauge degrees of freedom at codimension one, matter degrees of freedom at codimension two and Yukawa-type interactions at codimension three, in the spirit of the paradagm of model building in F-theory \cite{Vafa:2009se}. In order to see to what extent the new geometry we have obtained realizes all that, let us specify two unitary F-theory configurations and work with them throughout the rest of the section. We choose an $SU(4)$ and an $SU(5)$ model, since they display different properties of Yukawa couplings, as it will be clear in a moment.  Also, they are the lowest rank representatives of the even and odd unitary series with more familiar enhancements:  $SU(2)$ enhances to the Kodaira singularity III and $SU(3)$ to IV$^s$ on the O7-plane \cite{GM1}.

\subsubsection{SU(4)}
In order to identify the relevant objects at weak coupling, we have to study the behavior of the discriminant of the elliptic fibration as $\epsilon$ goes to $0$. Let us do that for both the Donagi-Wijnholt limit \eqref{DWLimitNonAbelian} and the new limit \eqref{SppLimitNonAbelian} and compare the two situations. 
 We use the Tate form of $SU(4)$ as expressed in table \ref{Table.TateForm} for a fiber of type $I^s_4$: 
That means $(a_1,a_2, a_3, a_4, a_6)$ have multiplicity $(0,1,2,2,4)$ along $\sigma=0$
\begin{equation}
a_2=\sigma a_{2,1}, \quad a_3=a_{3,2}\sigma^2, \quad a_4=a_{4,2} \sigma^2, \quad a_6=a_{6,4}\sigma^4.
\end{equation}
By applying the Donagi-Wijnholt limit, we obtain
\begin{equation}\label{LeadingSen}
\Delta|_{\rm DW}\sim\left[a_1^2+4\sigma a_{2,1}\right]^2\,\sigma^4\,\left[a_{4,2}^2+a_1(a_{3,2}a_{4,2}-a_1a_{6,4})-\sigma a_{2,1}b_{6,4}\right]\;\epsilon^2\,,
\end{equation}
The discriminant is factorized into three pieces whose vanishing respectively represents the O7-plane, the D7-stack hosting the $SU(4)$ gauge group and the Whitney umbrella D7-brane. We have two relevant matter curves here\footnote{To be more precise, one has to look at the full discriminant, which has the form $\sigma^4\,I_1$, the latter factor being a recombined 7-brane with $U(1)$ gauge group, responsible for canceling the tadpole. On $\sigma=0$, $I_1$ factorized into two branches which are the matter curves discussed above.\label{FullDelta}}:  
\begin{equation}
{\rm\bf 6}\;:\quad\begin{cases}\sigma=0\\ a_1=0\end{cases}\qquad,\qquad{\rm\bf 4}\;:\quad\begin{cases}
\sigma=0\\ a_{4,2}^2+a_1(a_{3,2}a_{4,2}-a_1a_{6,4})=0\,.\end{cases}
\end{equation}
 This is also consistent with the result of \cite{GM1}.
The first hosts matter fields transforming in the {\bf 6}, the antisymmetric representation of $SU(4)$, which originates from the symmetry enhancement to the $SO(7)$ group\footnote{This enhancement is slightly more generic than the $SO(8)$ we would expect from string theory, as the latter would require an additional factorization condition, which kills the monodromies. As usual, the ${\rm\bf 6}$ of $SU(4)$ arises from the decomposition of the adjoint of $SO(7)$, i.e. ${\rm\bf 21}={\rm\bf 15}+{\rm\bf 6}$.}. The second accommodates matter fields transforming in the {\bf 4}, the fundamental representation of $SU(4)$, which originates from the symmetry enhancement to $SU(5)$.  These two matter curves intersect in points of the D7-stack worldvolume where the further symmetry enhancement  is expected to occur to accommodate the coupling
\begin{equation}\label{YukawaSU4}
{\rm\bf 6}\,\bar{\rm\bf 4}\,\bar{\rm\bf 4}:\begin{cases}\sigma=0\\ a_1=0\\  a_{4,2}=0 \end{cases}
\end{equation}
Actually, being \emph{all} the enhancements of $SU(4)$ to either the unitary or the orthogonal type of group, we should have good chances of finding an effective description of this physics within the realm of weakly coupled type IIB string theory, possibly including D-instanton effects in order to reproduce certain perturbatively forbidden Yukawa interactions. The situation is different for $SU(5)$, as we will see shortly.

Let us now use the new limit \eqref{SppLimitNonAbelian} and expand the discriminant accordingly
\begin{equation}\label{LeadingSpp}
\Delta|_{\rm spp}\sim  \left[a_1^2+\sigma^2a_{2,2}\right]^2\,\sigma^4\,\left[a_{4,2}^2+a_1(a_{3,2}a_{4,2}-a_1a_{6,4})-\sigma^2a_{2,2}b_{6,4}\right]\;\epsilon^2\,,
\end{equation}
As one immediately sees, all the relevant features of the $SU(4)$ model are kept intact, since the pattern of intersections and enhancements are unchanged. However, what we should really be looking at is the discriminant after the blow-up of $B_3$, namely the proper transform of \eqref{LeadingSpp}. Recall that the proper transform of the polynomial defining the D7-stack $\sigma$ is $s$ and the class of the divisor $\{s=0\}$ is $\mathcal{D}-E$ with $\mathcal{D}$ the class of $\{\sigma=0\}$ and $E$ the exceptional class. The proper transform of $a_1$ is $a$, of class given by \eqref{nonCrepant}.  Thus we have 
\begin{equation}\label{LeadingSPPpt}
\hat{\Delta}|_{\rm spp}\sim  \left[a^2+s^2a_{2,2}\right]^2\,s^4\,\left[a_{4,2}^2+a w(a_{3,2}a_{4,2}-a w a_{6,4})-w^2s^2a_{2,2}b_{6,4}\right]\;\epsilon^2\,,
\end{equation}
Notice that we are providing here a smooth target space for type IIB strings at weak coupling: This is the resolved Calabi-Yau threefold defined in \eqref{resSPP}. This manifold projects onto the base \eqref{resolvedBase}, which is  \emph{different} from the base of the elliptic fibration we started with (it is connected to it by blow-down). As a consequence, \eqref{LeadingSPPpt} is not to be regarded as the discriminant of a Calabi-Yau elliptic fibration. Indeed, the $U(1)$ D7-brane which is not touched by the blow-up and is supposed to cancel the total charge of O7-plane and D7-stack (last piece in \eqref{LeadingSPPpt}) has no longer the right degree to do that. 
This means that to preserve the D7 tadpole in the type IIB theory on the resolved $\tilde{X}$, the spectrum of branes is not given by the proper transform of the discriminant. There is an  additional contribution to the D7  charge  required to satisfy the  D7 tadpole. This contribution is 
$4 E$. Since the D7 tadpole is equivalent to the Calabi-Yau condition for an elliptic fibration that would admit $\tilde{X}$ as its weak coupling limit, we can also consider the following scenario.

 We will now construct a Calabi-Yau elliptic fibration  with  base $\tilde{B}_3$. We also impose that at weak coupling it admits the double cover $\tilde{X}$ as in equation \eqref{resSPP} and  a $\SU(4)$ stack  over $\tilde{\mathcal{D}}:s=0$, the proper transform of $\mathcal{D}$.  However the $U(1)$ brane will  not coincide with the one in \eqref{LeadingSpp}. But it will have the proper degree to ensure that the elliptic fibration is Calabi-Yau and therefore automatically satisfies the  D7  tadpole.   In order to impose the proper stack, we use the Tate form for a fiber $I^s_4$ over the divisor  $\tilde{\mathcal D}=\mathcal{D}-E: s=0$. 
The coefficients of the Tate form are given by:
\begin{equation}
\tilde{a}_1=\tilde a_{1,0}, \quad \tilde{a}_2 =\tilde a_{2,1} s + \tilde a_{2,2} s^2 , \quad \tilde{a}_3= \tilde a_{3,2} s^2, \quad\tilde{a}_4=\tilde a_{4,2} s^2, \quad \tilde{a}_6=\tilde a_{6,4} s^4,
\end{equation}
where $\tilde{a}_p$ is by definition a section of a line bundle of class $p c_1(\tilde{B}_3)=p(c_1-E)$ and therefore  $a_{p,q}$ is a section of  a line bundle of class $p c_1(\tilde{B}_3)-q\tilde{\mathcal{D}}=p c_1-q \mathcal{D}-(p-q)E$. The coefficient $a_2$ has a deformation $a_{2,2}$ compatible with a fiber $I^s_4$ and useful to define a suspended pinch point weak coupling limit. If we take $a_{2,1}$ identically vanishing, the DW-weak coupling limit will coincide with the suspended pinch point weak coupling limit. If we keep $\tilde a_{2,1}$, then we will rescale it as $\tilde a_{2,1}\rightarrow \epsilon \tilde a_{2,1}$ in the weak coupling limit. 
Both ways, we get a spp  $\xi^2= \tilde a_1^2 + 4 \tilde a_{2,2}s^2$ in the weak coupling limit. It will coincide with the double cover of the resolved $\tilde{X}$ if we impose $\tilde a_1=a$ and $\tilde a_{2,2}=a_{2,2}$.  
The other Tate coefficients can be realized as follows:
 \begin{equation}\label{PT}
\begin{array}{lllll}
\ a_1&\longrightarrow&a&&c_1-E\\ \\
 a_{2,1}&\longrightarrow& s\, a_{2,1}^{(1,0)} + a\, a_{2,1}^{(0,1)}  & &2 c_1-\mathcal{D}-E\\ \\
a_{2,2}&\longrightarrow&a_{2,2}& &2 c_1-2\mathcal{D}\\ \\
a_{3,2}&\longrightarrow&s\,a_{3,2}^{(1,0)}+a\,a_{3,2}^{(0,1)}&&3 
  c_1-2\mathcal{D}-E\\  \\
a_{4,2}&\longrightarrow&s^2\,a_{4,2}^{(2,0)}+s\,a\,a_{4,2}^{(1,1)}+a^2\,a_{4,2}^{(0,2)}&&4  c_1-2\mathcal{D}-2E \\ \\ a_{6,4}&\longrightarrow&s^2\,a_{6,4}^{(2,0)}+s\,a\,a_{6,4}^{(1,1)}+a^2\,a_{6,4}^{(0,2)}&&6 c_1-4\mathcal{D}-2E\;,  
\end{array}
\end{equation}
where the last column indicates the divisor class of $\tilde a_{p,q}$. 
The superscript $(n,m)$ just means that we have $n$ powers of $s$ and $m$ powers of $a$ in front of that coefficient. 
Depending on the details of the model some of them may identically vanish.

By taking the DW limit of this new fibration we get the following discriminant at leading order:
\begin{equation}\label{LeadingSppPT}
\Delta|_{\rm spp}\sim\left[a^2+s^2a_{2,2}\right]^2\,s^4\,\left\{a^4\,\left[(a^{(0,2)}_{4,2})^2+a^{(0,1)}_{3,2}a^{(0,2)}_{4,2}-a^{(0,2)}_{6,4}\right]+\mathcal{O}(s)\right\}\;\epsilon^2\,.
\end{equation}
The last bracket in \eqref{LeadingSppPT} defines the recombined, tadpole-canceling $U(1)$ D7-brane.  The geometry of the Calabi-Yau uplift $\tilde{Y}$
 successfully reproduces the $SU(4)$ gauge degrees of freedom on the D7-stack at $s=0$. Let us now go to codimension two. From \eqref{LeadingSppPT} we deduce the two matter curves
\begin{equation}
\cancel{{\rm\bf 6}\;:\quad\begin{cases}s=0\\ a=0\end{cases}}\qquad,\qquad{\rm\bf 4}\;:\quad\begin{cases}s=0\\ (a^{(0,2)}_{4,2})^2+a^{(0,1)}_{3,2}a^{(0,2)}_{4,2}-a^{(0,2)}_{6,4}=0\,.\end{cases}
\end{equation}
 We obtain again a matter curve accommodating fields in the fundamental representation ${\rm\bf 4}$.   However,  matter arising from the intersection of the D7-stack with the O7-plane loses its support. This curve of intersection is in fact precisely the center of our blow-up. 
 This is an issue for those models, like F-theory inspired GUTs \cite{Donagi:2008ca, Beasley:2008dc}, which  require such a curve  for phenomenological reasons. We do not have a convincing solution yet.

One may think that the ${\rm\bf 6}$ is ``higgsed'' and split in two parts. Fix a divisor which intersects both the O7-plane and the D7-stack: The two branches of the ${\rm\bf 6}$-matter curve would then arise from the fixed divisor intersecting the O7-plane on one hand and the D7-stack on the other. This would then imply having a 7-brane source wrapping the fixed divisor, which from the structure of the discriminant \eqref{LeadingSppPT} is generically not the case for $SU(4)$. 
Therefore, one is led to constrain the complex structure of the fourfold in order to achieve this further factorization of the discriminant. Note, however, that a natural candidate divisor which in our resolved geometry interpolates between the D7-stack and the O7-plane is the exceptional divisor \eqref{ExceptionalSpp}. One can easily work out the same analysis for $SU(2k)$ with $k>2$, and realize that for $k\geq5$ factors of $w$ start appearing  in eq. \eqref{LeadingSppPT}. They are needed for consistency with 7-brane tadpole cancellation, expressed by the fundamental relation 
\begin{equation}\label{7-braneTadpole}
[\Delta]_{\rm spp}=12\,c_1(\tilde{B}_3)=12\,(c_1(B_3)-E)\,,
\end{equation}
where we used eq. \eqref{nonCrepant} and $[\Delta]$ means the divisor class of the discriminant. This provides, at least for high rank gauge groups, a natural playground for putting on solid base the conjecture of the ${\rm\bf 6}$-matter curve higgsing. We have not explored the details of this idea yet. We hope to come back soon to the problem in a future publication. Meanwhile, in section \ref{secondApproach}, we will propose a slight modification of limit \eqref{SppLimitNonAbelian}, which, though having a smaller range of validity, does not affect the matter curve in question.

\subsubsection{SU(5)}
Let us now come to the odd series of unitary groups and focus on the basic case of $SU(5)$.

By applying the Donagi-Wijnholt limit \eqref{DWLimitNonAbelian} to the discriminant, we obtain
\begin{equation}\label{LeadingSenOdd}
\Delta|_{\rm DW}\sim\left[a_1^2+4\sigma a_{2,1}\right]^2\,\sigma^5\,\left[a_1(a_{3,2}a_{4,3}-a_1a_{6,5})-a_{2,1}a_{3,2}^2+\mathcal{O}(\sigma)\right]\;\epsilon^2\,,
\end{equation}
The relevant matter curves now are (same observation as in footnote \ref{FullDelta} applies here)
\begin{equation}
{\rm\bf 10}\;:\quad
\begin{cases}
\sigma=0\\ 
a_1=0
\end{cases}
\qquad,\qquad{\rm\bf 5}\;:\quad 
\begin{cases}
\sigma=0\\ 
a_1(a_{3,2}a_{4,3}-a_1a_{6,5})-a_{2,1}a_{3,2}^2=0
\end{cases}
\end{equation}
Matter in the antisymmetric representation ${\rm\bf 10}$ arises from the enhancement to $SO(10)$ along the O7-plane, while matter in the fundamental representation ${\rm\bf 5}$ arises from the enhancement to $SU(6)$ along the remainder, invariant D7-brane. Here we readily see that there are two kinds of Yukawa couplings arising from the intersection matter curves\cite{Marsano:2011hv}
\begin{equation}\label{YukawaSU5}
{\rm\bf 10}\,\bar{\rm\bf 5}\,\bar{\rm\bf 5}\;:\quad
\begin{cases}
\sigma=0\\ a_1=0\\ a_{3,2}=0
\end{cases}
\qquad,\qquad{\rm\bf 10}\,{\rm\bf 10}\,{\rm\bf 5}\;:
\begin{cases}
\sigma=0\\ 
a_1=0\\ 
a_{2,1}=0\,.
\end{cases}
\end{equation}
The first is analogous to the one for $SU(4)$ \eqref{YukawaSU4}, and it comes from the enhancement to $SO(12)$. The second is peculiar of $SU(5)$ and it is exactly localized on the conifold points of the singular geometry \eqref{ConifoldSingularity}. It comes from the enhancement to `E$_6$'.  Hence, the enhancement to an exceptional gauge group occurring for $SU(5)$ tells us that in this case the physics hidden in the conifold points may well be intrinsically strongly coupled and thus impossible to reproduce just by means of fundamental strings\footnote{One would need string-junctions, which are believed to be the `fundamental'' objects of F-theory.}.

If we now use the new limit \eqref{SppLimitNonAbelian} to expand the discriminant, we get
\begin{equation}\label{LeadingSppOdd}
\Delta|_{\rm spp}\sim\left[a_1^2+\sigma^2a_{2,2}\right]^2\,\sigma^5\,a_1\,\left[a_{3,2}a_{4,3}-a_1a_{6,5}+\mathcal{O}(\sigma)\right]\;\epsilon^2\,,
\end{equation}
where we see that the previous pattern of intersections and enhancements is not respected. The conifold points are scaled away, as it should be, while the $SO(12)$ points are kept. However, the meaningful quantity is the discriminant after the blow-up of the suspended pinch point. In the table \eqref{PT} we have to constrain the complex structure in order to extract a further factor of $s$ from $\tilde a_{4,2}$ and $\tilde a_{6,4}$, as required by the Tate prescription for $SU(5)$. This leads us to the following discriminant
\begin{equation}\label{LeadingSppPTOdd}
\Delta|_{\rm spp}\sim\left[a^2+s^2a_{2,2}\right]^2\,s^5\,a^3\,\left[a^{(0,1)}_{3,2}a^{(0,1)}_{4,3}-a^{(0,1)}_{6,5}+\mathcal{O}(s)\right]\;\epsilon^2\,,
\end{equation}
where now the last polynomial in square brackets has class $5c_1-5\mathcal{D}$.  Our new geometry successfully reproduces the $SU(5)$ gauge degrees of freedom on the D7-stack at $s=0$. As for codimension two, we have again
\begin{equation}
\cancel{{\rm\bf 10}\;:\quad\left\{\begin{array}{l}s=0\\ a=0\end{array}\right.}\qquad,\qquad{\rm\bf 5}\;:\quad\left\{\begin{array}{l}s=0\\ a^{(0,1)}_{3,2}a^{(0,1)}_{4,3}-a^{(0,1)}_{6,5}=0\,.\end{array}\right.
\end{equation}
While the curve accommodating matter in the fundamental representation is successfully reproduced, the one hosting matter in the antisymmetric disappears. One may think of ``higgsing'' the antisymmetric matter curve using the exceptional divisor, in analogy to what proposed for $SU(4)$. Here factors of $w$ appear in \eqref{LeadingSppPTOdd} for  $SU(2k+1)$ with $k\geq4$. However, we defer a more accurate analysis of this issue to future work.

\subsection{An alternative spp}\label{secondApproach}

 As stressed in the previous section, the suspended pinch point geometry after resolution does not  reproduce the antisymmetric matter curve. In this section we propose a way out of this problem, by slightly modifying the definition of the new weak coupling limit \eqref{SppLimitNonAbelian}. We focus our attention on $SU(4)$ F-theory configurations and only say few words about other cases towards the end of the section.

Assume the D7-stack wraps a \emph{spin} manifold. This hypothesis is necessary for this alternative limit to work. Then consider the following weak coupling limit
\begin{equation}\label{sppLimitSpin}
\begin{cases}
a_{2,1} & \longrightarrow \epsilon \,a_{2,1}+\tfrac14P^2\\ 
a_{3,2} &\longrightarrow \epsilon \,a_{3,2}\\ 
a_{4,2} &\longrightarrow \epsilon \,a_{4,2}\\ 
a_{6,4} &\longrightarrow \epsilon^2 \,a_{6,4}
 \end{cases}
\end{equation}
 where $P$ is a section of a line bundle of  class $c_1-\mathcal{D}/2$, which makes sense since $\mathcal{D}$ is by assumption an even class. The new Calabi-Yau theefold geometry is still singular of the suspended pinch point-type
\begin{equation}\label{SPPspin}
X_3:(\xi-a_1)\,(\xi+a_1)=\sigma\,P^2\,.
\end{equation}
But this time the curve of singularities is \emph{not} the intersection of the D7-stack with the O7-plane. This different improvement of the Donagi-Wijnholt limit is still harmless from the point of view of the discriminant, as we get
\begin{equation}\label{LeadingSppSpin}
\Delta|_{\rm spp}\sim\left[a_1^2+\sigma P^2\right]^2\,\sigma^4\,\left[a_{4,2}^2+a_1(a_{3,2}a_{4,2}-a_1a_{6,4})+\mathcal{O}(\sigma)\right]\;\epsilon^2\,.
\end{equation}
Let us now blow-up \eqref{SPPspin} and convince ourselves that limit \eqref{sppLimitSpin} successfully reproduces the physics of the $SU(4)$ models in codimension one and two. The resolution procedure goes exactly as in the previous case. The resolved Calabi-Yau threefold is the complete intersection
\begin{equation}\label{resSPPspin}
\tilde{X}_3:\left\{\begin{array}{rcl}\xi^2&=&a^2+\sigma\,p^2\\ w\,p&=&P \\ w\,a&=&a_1\,,\end{array}\right.
\end{equation}
and $\xi\,p\,a$ is an element of the Stanley-Reisner ideal of the ambient six-dimensional manifold. Again it can be viewed as the double cover of the manifold
\begin{equation}\label{resolvedBaseSpin}
\tilde{B}_3\;:\quad\left\{\begin{array}{rcl}w\,p&=&P \\ w\,a&=&a_1\,.\end{array}\right.\qquad,\qquad\begin{array}{ccc}p&a&w\\ \hline 1&1&-1\end{array}\,.
\end{equation}
which is the blow-up of $B_3$ along the curve $\{a_1=P=0\}$.

The stack of D7-branes and its orientifold image are described by the following systems
\begin{equation}
{\rm D7}_\pm:\left\{\begin{array}{rcl}\xi&=&\pm a\\ \sigma&=&0 \\ wp&=&P\\ wa&=&a_1\,.\end{array}\right.
\end{equation}
They now intersect on a curve which lies on the O7-plane. The latter is the surface
\begin{equation}\label{O-planeSpin}
{\rm O7}:\left\{\begin{array}{rcl} \xi&=&0\\ a^2&=&-\sigma \\ w\,a&=&a_1 \\ w&=&P\,,\end{array}\right.
\end{equation}
where we have fixed the gauge by taking $p=1$. We easily recognize from \eqref{O-planeSpin} a surface wrapping the divisor $\{a_1^2+\sigma P^2=0\}$ of the original base $B_3$, as it should be. Finally, the exceptional divisor is
\begin{equation}
E:\begin{cases}
w       &=0\\ 
a_1    &=0 \\
 P       &=0\\
 \xi^2 &=a^2+\sigma\,p^2
\end{cases}
\end{equation}
which has the geometry of an orientifold-invariant, quadratic $\mathbb{P}^1$ with homogeneous coordinates $a,p$, fibered over the locus $\{P=0\}\cap\{a_1=0\}\subset B_3$. On the location of the D7-stack, $\{\sigma=0\}$, the fiber of the exceptional divisor splits into two linear spheres, $\mathbb{P}^1_{pa}|_{\xi=\pm a}$, exchanged by the orientifold involution. This last geometry may turn useful, as it will be clear below.

Let us now look at the proper transform of the discriminant \eqref{LeadingSppSpin}. It is not difficult to understand that blue the Tate coefficients $a_{3,2}, a_{6,4}$ cannot be constrained as imposed by 7-brane tadpole cancellation. In contrast $a_{4,2}$ must be replaced by the monomial $p^4$, which is the most generic form of the appropriate degree. Therefore the discriminant simply reads
\begin{equation}\label{LeadingSppSpinPT}
\Delta|_{\rm spp}\sim \left[a^2+\sigma p^2\right]^2\,\sigma^4\,p^8\;\epsilon^2\,.
\end{equation}
Here we see that the U(1) D7-brane has undergone a drastic change and, due to its high degree, has given rise to a stack of eight D7-branes plus separated orientifold images (much like the $SU(4)$ stack in the resolved suspended pinch point geometry of section \ref{resCY3spp}). This new stack accommodates an $SU(8)$ flavor symmetry. There are two possible matter curves, described by the following intersections
\begin{equation}
{\rm\bf 6\;:\quad\left\{\begin{array}{l}\sigma=0\\ a=0\end{array}\right.}\qquad,\qquad{\rm\bf 4}\;:\quad\left\{\begin{array}{l}\sigma=0\\ p=0\,.\end{array}\right.
\end{equation}
The first is the curve where matter in the antisymmetric representation of $SU(4)$ is localized, which arises from the ordinary enhancement to $SO(8)$ along the O7-plane. Notice that it now survives the resolution. The second is the curve where matter in the fundamental representation of $SU(4)$ lives. Since this curve is the intersection of the gauge stack with an $SU(8)$ flavor stack, there is an enhancement to $SU(12)$ along it and the matter fields localized there transform in the fundamental representation of the flavor group. If we now look at the intersection of these two matter curves to search for the Yukawa couplings, we readily see that it is empty, because it is part of the locus which has been blown-up. However, we see that the triple intersection we are looking for is replaced by the curve $E\cap {\rm D7}_+\cup E\cap {\rm D7}_-$
\begin{equation}
\begin{cases}
w&=0\\
 a_1&=0 \\ 
P&=0\\
 \sigma&=0\\ 
\xi&=a\end{cases}
\quad
\cup \quad
\begin{cases}
w&=0\\ 
a_1&=0 \\ 
P&=0\\ 
\sigma&=0\\ 
\xi&=-a
\end{cases}
\end{equation}
whose typical fiber, as already mentioned, is a pair of $\mathbb{P}^1$s, one the orientifold image of the other, touching at a point. One is now tempted to argue that we have an effective, non-perturbative type IIB description of the ${\rm\bf 6}\,\bar{\rm\bf 4}\,\bar{\rm\bf 4}$ Yukawa coupling by means of D1-instantons wrapping one $\mathbb{P}^1$ and anti-D1-instantons wrapping the image $\mathbb{P}^1$. However, we have not performed an accurate analysis of this system: Besides proving that it is actually stable, one has to make sure that the instantons in question have the right number of neutral zero-modes so to contribute to the superpotential and generate the wanted Yukawa coupling. We hope to clarify all that in a future work.

To conclude this section, let us stress that limit \eqref{sppLimitSpin} does not properly work beyond $SU(4)$, i.e. for F-theory configurations with $SU(N\geq4$) singularity. This is because there is no way in the geometry of $\tilde{B}_3$ to satisfy the 7-brane tadpole. Therefore the validity of the weak coupling limit presented in this section is limited to $SU(4)$ F-theory configurations with gauge stack wrapping a spin manifold.

\section{Conclusions}\label{Outlook}

In this paper, we have discussed different realizations of  Sen weak coupling limit\cite{Sen.Orientifold}  which are alternatives or specializations of the  traditional Donagi-Wijnholt ansatz \cite{Donagi:2009ra}. 
The main purpose has been to provide a systematic way of solving the conifold problem afflicting the DW ansatz when applied to singular fibrations with unitary gauge groups: These singularities do not admit admissible crepant resolutions\footnote{An admissible crepant resolution of a double cover $X$ is a resolution  of $X$  that preserves the first Chern class and is compatible with the structure of the double cover (and therefore with the orientifold involution).}.  We also analyze the weak coupling limit of all  gauge groups implemented by Tate forms, including the exceptional ones. 
The properties of these limits are somehow surprising: 
\begin{enumerate}
\item  The gauge group seen at weak coupling  is not necessarily the same as the one observed in F-theory. This is expected in certain cases, like for example for exceptional gauge groups in F-theory as they are not present at weak coupling. The groups seen in the weak coupling limit can be orthogonal, unitary or symplectic. An orthogonal gauge group appears when the locus of the brane coincides with a component of the orientifold locus. 
\item The gauge group seen at weak coupling is not necessarily a subgroup of the F-theory group. For example, for $E_6$, we get a group $\SO(12)$. One could argue that this $\SO(12)$ should reduce to $\SO(10)$ if it is generated by a perturbative subset of  open strings that generate $E_6$ in F-theory. This would match the description of $E_6$ in F-theory using monodromies of ``$ABC$" branes. A clear, string-based  deduction of gauge group at weak coupling for exceptional singularities (II,III,IV and duals thereof) of elliptic fibration of dimension grater that two is still missing. This is due to the fact that there is no way of getting a weak string coupling on the would-be gauge stack\footnote{This problem does not happen for exceptional singular fibers in weak coupling limits considered in \cite{AE2,EFY} which are not based on  Tate forms.}. 
\item In the DW-weak coupling limit, the same gauge group is obtained for fibers  regardless of it being split/non-split/semi-split. This is true with the exception of the $I_n$ fibers that  lead to unitary and symplectic gauge groups. 
\item  For unitary gauge groups, we can maintain the group and its  rank and get a double cover with an  admissible crepant resolution  if we use  the  suspended pinch point (spp) limit. However, it  requires introducing a term ($a_{2,2}$) which is a section of a line bundle of class $2L-2\mathcal{D}$. The existence of such a section is a non-trivial  topological constraint.  There is an alternative limit also leading to a double cover with an admissible crepant resolution and which is free of such a topological constraint but which  leads to  orthogonal gauge groups at weak coupling. 

\item
We also note  a possible tension between a (crepant) resolution $\tilde{X}\rightarrow X$ of  a Calabi-Yau and the D7 tadpole cancellation condition which requires the vanishing of the total D7 charge. Indeed, after a resolution, the branes are expected to wrap the proper transforms of the cycles they used to wrap in $X$. However, if some of these cycles intersect the center of the blow-up with multiplicities, their classes will get a contribution from the exceptional divisor. It follows that the D7 tadpole can be in  jeopardy as it is based on a delicate equilibrium between the class of the D7 branes and the orientifold plane. We will give an important example below. 
\begin{example}[Tadpole requirements for the typical configuration]
Consider the typical situation that occurs for the weak coupling limits we have in this paper:  A singular double cover is resolved by blowing-up a codimension-two locus of multiplicity 2. The spectrum  consists of an orientifold $\underline{O}$, a stack of  $r$ D7 branes on $\underline{D}$ and a spectator $U(1)$ brane $\underline{D}'$, all described in the base. We assume that 
the stack intersects the center of the blow-up  with multiplicity 2.  Before the resolution, we have the tadpole 
$$8[\underline{O}]-(r [\underline{D}]+[\underline{D}'])=0.$$ 
After the blow-up, we can evaluate the D7 charge mismatch:
$$
8[\underline{O}-E]-(r [\underline{D}-E]+[\underline{D}'])=(r-8)E.
$$
If $r=8$,  the proper transform of the spectrum does satisfy the tadpole. 
If $r<8$, the tadpole charge would require  a negative contribution proportional to the exceptional divisor $E$. This can be for example a stack of $8-r$ anti-D7 branes, which will break supersymmetry.  If $r>8$, the tadpole can be canceled by wrapping $(r-8)$ D7-branes on the exceptional divisor. We present a solution to the problem when $r<8$. Indeed, we cannot keep the spectrum of the proper transforms. But if we would like to keep the orientifold and the stack unchanged we can modify the remaining brane in such a way that the tadpole is preserved. We can think of it as a supersymmetric brane recombination of the stack of anti-branes and the witness  brane.  We have obtained a natural description of  the final result of such a recombination using a Calabi-Yau elliptic fibration  over the base of the resolution of $X$.
\end{example}
\end{enumerate}

In the second part of the paper, we have focused on questions relevant for  phenomenological GUT model building.
 In particular we have explored the possibility of realizing in an effective way as much as we could of the physics of F-theory $SU(N)$ configurations using only the weakly coupled dynamics of type IIB strings. Therefore we started by requiring that we have (on a arbitrary divisor $\{\sigma=0\}$) an $SU(N)$ stack of D7-branes and its orientifold image in a smooth Calabi-Yau threefold. This one condition (together with the fact that the Calabi-Yau threefold is the double cover of the base of the elliptic fibration) already constrain the hypersurface equation to have the following conifold form
\begin{equation}
\xi^2=a_1^2+\sigma\,B\,,
\end{equation}
where $B$ is a polynomial of the base of the appropriate degree. $B$ is the only factor we can play with in order to achieve a more tractable singularity. Now, if we impose that $\sigma$ divides $B$, as we did in subsection \ref{PhysicalProperties}, the antisymmetric matter curve becomes the singular locus and it is blown-up in the resolved picture. Alternatively, we can deform $B$ to be a perfect power. Here two sub-cases are possible. The power is even (the basic case of power two has been explored in subsection \ref{secondApproach}), which only works with the assumption of spin-ness of the gauge divisor; The case of $SU(4)$ seems to work with this strategy, but higher rank gauge groups seems incompatible with 7-brane tadpole cancellation. The power is odd; This case reduces to the original conifold, after a series of resolutions.

More work is needed in the investigation of a full effective description within the realm of type IIB string theory of the strongly coupled physics of unitary F-theory configurations. In particular, suitable instanton effects will be required in order to reproduce certain expected Yukawa interactions. We hope to come back to all these issues in a future work.

\subsection*{Acknowledgements}
We would like to thank Andr\'es Collinucci for initial collaboration and for many fruitful conversations. We also like to acknowledge useful discussions with Frederik Denef, I\~naki Garc\'ia-Etxebarria, Thomas Grimm, Hirotaka Hayashi, Stefan Hohenegger, Shamit Kachru, Timo Weigand, Martijn Wijnholt and Shing-Tung Yau. 
We would  like to acknowledge the hospitality of the Simons Center  for Geometry and Physics in Stony Brook where this project was born.  M.E.  is very grateful to the members of the Taida Institute  for  hospitality. He would also like to thank Imran Esole for his joyful cooperation at different stage of this project.

\appendix

\section{Admissible crepant resolutions of double covers}\label{AdmissibleResolutions}

In this appendix, we will study some standard properties of double covers and  their resolutions. We refer to \cite{CS,Miranda} for the proof. 
We will first start by recalling some basic definitions necessary for the rest of the discussion. 
\begin{definition}[Finite maps]
We denote by $A(X)$ the coordinate ring of a variety $X$. 
A map $X\rightarrow Y$ is said to be {\em finite} if $A(X)$ is  locally  a finitely generated  module over $A(Y)$. 
The   {\em degree} of the finite map is then by definition the degree $[A(X):A(Y)]$ of the field extension . If $X$ and $Y$ are projective, a finite map is equivalent to a map with finite fibers. 
\end{definition}

We can now introduce the definition of a double cover. 
\begin{definition}[double cover]
A double cover $\rho: X\rightarrow B$ is a flat finite map or rank 2 between  varieties $X$ and $B$. 
\end{definition}
The flatness condition is to ensure that each fiber has the same number of points (counted with multiplicity). 
Locally, over an open affine set $U\in B$, a double cover  is given by an equation of the type $z^2=h_U$. The collection of such $h_{\alpha}$ over a finite open cover $\bigcup_{\alpha } U_\alpha$ of the base $B$ defines the   {\em branch divisor}  $\underline{O}$ of the double cover using the local equations  $h_\alpha=0$. Such a divisor $\underline{O}$ is uniquely determined by the double cover map. This can be seen  for example\footnote{See  Eisenbud and Harris's book on intersection theory.}   by computing the pushforward of the relative differential sheaf $\Omega_{X/B}$. The variable $z$ is then a section of $\mathscr{L}$ and  $\underline{O}$ is a section of  $\mathscr{L}^2$.  Together,  the pair $(\underline{O}, \mathscr{L})$ characterizes the double cover $\rho: X\rightarrow B$. This is summarized in the following proposition. 

\begin{prop}[Characterization of the space of double covers over a fixed base (\cite{CS,Miranda})]
Let $B$ be any variety over a field of characteristic different than $2$. There is a one-to-one correspondence between double covers $\rho:X\rightarrow B$ and 
pairs $(\mathscr{L}, \underline{O})$ consisting of a line bundle $\mathscr{L}$ and the divisor $\underline{O}\in \mathscr{L}^{2}$.  
\end{prop} 

The next proposition characterizes $\mathscr{O}_X$ and its multiplication structure in terms of the algebraic properties of the base and the pair $(\underline{O}, \mathscr{L})$.

\begin{prop}[Algebraic properties of double covers]
Consider $\rho:X\rightarrow B$ a finite map of degree two of a variety $X$ onto a smooth variety $B$. We denote by $\underline{O}\subset B$ the branch locus of the map $\rho$. We then have the following properties:
\begin{enumerate}[(i)]
\item $\rho_\star \mathscr{O}_X=\mathscr{O}_B\oplus \mathscr{L}$ for a line bundle $\mathscr{L}$ on $B$. 
\item The multiplication in $\mathscr{O}_X$  is given by a map $\mathscr{L}\otimes \mathscr{L}\rightarrow \mathscr{O}_B$  or equivalently , by a section $h\in H^0(B, \mathscr{L}^{2})$. The zero locus of $h$ is the branch locus of $B$. Hence $\mathscr{O}_Y(\underline{O})\simeq \mathscr{L}^{2}$. 
\end{enumerate}
\end{prop}

We now explain how the smoothness of a double cover is equivalent to the smoothness of its branch divisor $\underline{O}$.

\begin{prop}[Characterization of smooth double covers] \hfill
\begin{enumerate} [(i)]
\item A double cover $X$ is smooth if and only if  its branch divisor $\underline{O}\subset B$ is smooth.
\item If the branch divisor $\underline{O}$ is smooth than  $K_X=\rho^\star (K_B\otimes \mathscr{L}^{-1})$,  where $K_X$ and  $K_B$ are  respectively the canonical classes of the corresponding varieties $X$ and  $B$. 
\end{enumerate}
\end{prop}
The previous proposition implies that we can reduce the resolution of a double cover to the resolution of its branch divisor. 
This fact is exploited in the following proposition which describes how a double cover behaves under a blow-up in the base. The idea is to pull back the double cover on the  blown-up base. 
The pull-back will use a fibered product so we will first recall its definition.  Consider two maps $\varphi_1 :X_1\rightarrow S$ and  $\varphi_2 :X_2\rightarrow S$. Then $X_1\times_S X_2$ is the set of pairs $(x_1,x_2)\in X_1\times X_2$ that project to the same element on $S$:
 $X_1\times_S X_2=\{(x_1,x_2)\in X_1\times X_2: \  \varphi_1(x_1)=\varphi_2(x_2)\}$ 

\begin{prop}[Resolution of double covers (see for example \cite{CS}]
Let $Z$ be a smooth sub-variety of $B$, we denote by  $j: B_Z\rightarrow B$ the blow-up of $B$ along $Z$ and  $E_Z\subset B_Z$ the exceptional divisor of the blow-up. Let $X_Z=X\times_B B_Z$, then  $\rho_Z:\tilde{X}_Z\rightarrow B_Z$ with $\tilde{X}_Z$ is the normalization of $X_Z$. The  branch locus  $\tilde{\underline{O}}_Z$ of $\rho_Z$ is  
\begin{equation}
\tilde{\underline{O}}_Z=\underline{O}_Z+\epsilon_Z E_Z, \quad \text{with}\quad \epsilon_Z=\begin{cases}
0\quad  \text{if  the multiplicity of $Z$ in $\underline{O}$ is even  },  \\
1 \quad  \text{if  the multiplicity of $Z$ in $\underline{O}$ is odd  },
\end{cases}
\end{equation}
where $\underline{O}_Z$ is the proper transform of $\underline{O}$ in $B_Z$. 
\end{prop}

\begin{prop}
A blow-up of the base $j:B_Z\rightarrow B$ along a smooth sub-variety $Z$ of codimension $r$ preserves the canonical class of a double cover  
 if and only if the multiplicity $m$ of $Z$ along the branch cover $\underline{O}$ is $m=2r-2$ or $m=2r-1$.
\end{prop}
\begin{proof}
A direct application of  the familiar formula for the canonical class after a blow-up with exceptional divisor $E_Z$ gives 
$$ K_{B_Z}=j^\star K_B+(r-1) E_Z, \quad K_{\underline{\tilde{O}}_Z}=j^\star K_{\underline{O}}-(m-\epsilon_Z) E_Z.$$ 
It follows that 
\begin{equation}
K_{\tilde{X}_Z}=K_{B_Z}+\frac{1}{2}K_{\tilde{\underline{O}}_Z}=j^\star(K_B+\frac{1}{2}K_{\underline{O}})+(r-1-\frac{m-\epsilon_Z}{2})E_Z,
\end{equation}
We see from this formula that the condition to preserve the canonical class is the vanishing of the coefficient in front of $E_Z$, this is equivalent to 
$m-\epsilon_Z+2=2r$ and the only solutions of this equation are  $m=2(r-1)$ and $m=2(r-1)+1=2r-1$. 
\end{proof}
As a direct application, let us see under which conditions the blow-up of a smooth sub-variety  $Z$ of codimension 1, 2 or 3 preserves the canonical class of the double cover:
\begin{lemma}The double cover of a  blow-up of $B$ along a  subvariety $Z\subset B$ or rank $r\leq 3$ will preserve the canonical class of the double cover of $B$ in the following cases:
\begin{enumerate}
\item $Z$ is  of codimension  $r=1$ in $B$  and  multiplicity $m=0$ or $m=1$ along $\underline{O}$.  
\item $Z$ is  of codimension  $r=2$ in $B$  and  multiplicity $m=2$ or $m=3$ along $\underline{O}$.
\item $Z$ is  of codimension  $r=3$ in $B$  and  multiplicity $m=4$ or $m=5$ along $\underline{O}$.
\end{enumerate}
 \end{lemma}

\subsection{Resolution of $t^2=xy$}
We consider a  double cover $\rho:X\rightarrow B$ with a branch divisor which has a $A_1$ singularity along a codimension-2 locus. The equation of such a double cover is 
the quadric cone 
\begin{equation}
 X:t^2=xy.
\end{equation}
The codimension-2 singular locus is  $t=x=y=0$.  
Rewriting the defining equation of $X$ as the following rational relation 
\begin{equation}
\frac{t}{y}= \frac{x}{t},
\end{equation}
we can easily find a crepant  resolution by introducing a $\mathbb{P}^1$ with projective coordinates $[\alpha:\beta]$ and imposing the rational relation   $x/t=\alpha/ \beta$. This leads to the crepant  resolution 
\begin{equation}
\tilde{X}
\begin{cases}
\alpha t -\beta x=0\\
\alpha y -\beta t =0
\end{cases} 
\end{equation}
written in  the ambient space $X\times \mathbb{P}^1$. 
It is easy to check that $\tilde{X}\rightarrow X$  is indeed a resolution of singularities:
(i)  $\tilde{X}$ is smooth since its Jacobian has rank 3.
 (ii) $\tilde{X}\rightarrow X$ is defined by the projection $(x,y,t)\times[\alpha:\gamma]\mapsto (x,y,t)$. It is trivially a birational transformation.  (iii)The exceptional locus is a $\mathbb{P}^1$ fibration over $x=y=t=0$. 
(iv) The projection $\tilde{X}\rightarrow X$ is an isomorphism  away from the codimension-2 locus $x=y=t=0$ of $X$ and its inverse image in $\tilde{X}$. 

\subsubsection{Involution}
The resolution $\tilde{X}$ admits an involution that reduces to the involution of the double cover $X$.  It is induced by the following involution of the ambient space:
\begin{equation}(t,x,y,[\alpha:\beta])\mapsto (-t,x,y,[\alpha:-\beta)].\end{equation}
In the ambient space, the fixed locus is $t=\beta\alpha=0$, the union of two non-intersecting  codimension-2 subvarieties $t=\alpha=0$ and $t=\beta=0$. When restricted to  $\tilde{X}$, it reduces to  the union of  two non-intersecting divisors: $t=\beta=y=0$ and $t=\alpha=x=0$.
 
\subsection{Crepant resolution of $t^2=x^2 z -y^2$.}\label{AppSpp}
The double cover $X:t^2=x^2 z -y^2$ is singular along the codimension-2 locus $t=x=y=0$ of multiplicity 2. The singularity worsen in codimension-3  at $t=x=y=z=0$. 
To resolve the variety $t^2=x^2 z -y^2$, we rewrite it as $t^2+y^2=x^2 z$ and after a change of variable $u_\pm=y\pm t \mathrm{ i}$ (where $\mathrm{i}^2=-1$), it becomes the normal equation of the  {\em suspended pinch point}:
$$
u_+ u_- =x^2 z.
$$
We can write is as the rational relation $\frac{u_+}{x}\frac{u_-}{x}=z$.  We introduce two $\mathbb{P}^1$s parametrized by $[\alpha_\pm:\beta_\pm]$ and  we take $u_\pm/x=\alpha_\pm/\beta_\pm$. This  respects the involution $t\mapsto -t$ (or equivalently $u_\pm\mapsto u_\mp$). The crepant resolution is  
\begin{equation}
\tilde{X}
\begin{cases}
\beta_+ u_+ -x\alpha_+ &=0\\
\beta_- u_- -x\alpha_- &=0\\
\alpha_+ \alpha_-  \  - z\beta_+ \beta_- &=0 
\end{cases}
\end{equation}
Over $x=y=t=0$, we have a $\mathbb{P}^1$ given by a quadric in $\mathbb{P}^1\times\mathbb{P}^1$. It  enhances to two transversally intersecting  $\mathbb{P}^1$  (given by a line in each of the two  rulings of   $\mathbb{P}^1\times\mathbb{P}^1$)  as the quadric degenerates over  $x=y=t=z=0$.
The involution is 
$$(t,x,y,z,[\alpha_\pm:\beta_\pm])\mapsto (-t,x,y,z,[\alpha_\mp :\beta_\mp]).$$
The fixed locus is  
$t=\alpha_+ -\alpha_- =\beta_+ - \beta_- =x^2z -y^2 =\alpha_+^2-z \beta_+^2=0$.

\subsection{Resolution of $t^2= x z+y^2$.}\label{AppConifold}
To resolve the variety $X:t^2=x z -y^2$, we rewrite it as $t^2+y^2=x z$ and after a change of variable $u_\pm =y\pm t $, it takes the normal form of the conifold
$$
X: \quad u_+ u_- =x z.
$$
We can write is as $\frac{u_+}{x}=\frac{z}{u_-}$.  We introduce a $\mathbb{P}^1$s parametrized by $[\alpha:\beta]$ and satisfying  $u_+/x=\alpha/\beta$. This leads to  the resolution 
\begin{equation}
\tilde{X}
\begin{cases}
 \beta  u_+  - \alpha x &=0\\
\beta  z     \  \   - \alpha u_- &=0
\end{cases}
\end{equation}
However this small resolution does not respect the involution of $X$.

\section{The type IIB Calabi-Yau threefold: Quadric cone case}\label{BlowUpQuadricCone}

In this section we want to present a resolution of the singular double cover which admits the singularity of a quadric cone as in eq. \eqref{QuadricConeSing}, namely 
$$X: \xi^2=\sigma a_{2,1}.$$
Like in section \ref{MathematicalProperties}, we will use toric methods as toric notations are very popular with physicists working on model building.

The resolved geometry is specified by the system
\begin{equation}
\tilde{X}_3\;:\quad
\begin{cases}
\xi^2&=s \,a\\ 
w\,s&=\sigma \\ 
w\,a&=a_{2,1}\end{cases}
\qquad \begin{array}{cccc}s&a&\xi&w\\ \hline 1&1&1&-1\end{array}
\end{equation}
where on the right the ``exceptional'' projective weight is displayed. This is now a perfectly smooth Calabi-Yau threefold, still invariant under the orientifold involution\footnote{In this case one can alternatively define the orientifold involution by reversing the sign of $s,a,w$ at the same time. This is clearly gauge-equivalent to sending $\xi\to-\xi$.}. Let us study some of the properties of this new geometry. 

The two branches of the O7-plane $\xi=0$ look as follows
\begin{equation}
{\rm O7}:
\begin{cases}\xi&=0\\ 
s&=0 \\ 
\sigma&=0\\ 
w&=a_{2,1}\,
\end{cases}
\qquad\cup\qquad
\begin{cases}
\xi &=0\\ 
a &=0 \\
 a_{2,1}&=0\\ 
w  &= \sigma
\end{cases}
\end{equation}
They do not intersect any more. The first branch, being also the locus where the original stack of $N$ D7-branes sits (together with its image), hosts an $SO(2N)$ gauge theory. The former locus of singularities, namely the intersection of the two branches, is now replaced by the exceptional divisor
\begin{equation}
E:\begin{cases}
w&=0\\
 a_{2,1}&=0 \\
 \sigma&=0\\
 \xi^2&=s\,a
\end{cases}
\end{equation}
which has the geometry of an orientifold-invariant, quadratic $\mathbb{P}^1$ with coordinates $a,s$, fibered over the locus $\{\sigma=0\}\cap\{a_{2,1}=0\}\subset B_3$.  As was the case in section \ref{resCY3spp}, also here the resolved Calabi-Yau threefold can be seen as the double cover of the resolved base $\tilde{B}_3$, the latter being the blow-up of $B_3$ along the curve $\{\sigma=a_{2,1}=0\}$.

\end{document}